\documentclass[11pt,draftcls,onecolumn]{IEEEtran}

\usepackage{amsmath}
\usepackage{amssymb}
\usepackage{graphicx}
\usepackage[noadjust]{cite}
\usepackage{mathrsfs}
\usepackage{stmaryrd}
\usepackage{hyperref}
\usepackage{subfigure}
\usepackage{color}
\newtheorem{theorem}{Theorem}
\newtheorem{corollary}{Corollary}[theorem]
\newtheorem{lemma}[theorem]{Lemma}

\newtheorem{definition}{Definition}

\newtheorem{remark}{Remark}
\newtheorem{example}{Example}

\newcommand{\set}[1]{\mathcal{#1}}
\renewcommand{\sf}[1]{\mathsf{#1}}
\newcommand{\bm}[1]{\boldsymbol{#1}}
\newcommand{\markov}{\textnormal{\mbox{$\multimap\hspace{-0.73ex}-\hspace{-2ex}-$ }}}
\newcommand{\cln}[3]{(#1\succeq #2\ |\ #3)}
\newcommand{\step}[2]{\stackrel{\textnormal{(#1)}}{#2}}

\allowdisplaybreaks[4] 

\graphicspath{{.}{./Figures/}}

\begin{document}

\title{Source Coding Problems with Conditionally Less Noisy Side Information}

\author{Roy Timo, Tobias J. Oechtering and Mich\`{e}le Wigger 
\thanks{R. Timo is a Research Fellow with the Institute for Telecommunications Research at the University of South Australia. R.~Timo was a visiting Associate Research Scholar with the Department of Electrical Engineering at Princeton University while the work in this paper was undertaken (e-mail: roy.timo@unisa.edu.au, rtimo@princeton.edu). R. Timo is supported by the Australian Research Council Discovery Grant DP120102123.}

\thanks{Tobias J. Oechtering is with the ACCESS Linnaeus Center, KTH Royal Institute of Technology (e-mail: oech@kth.se).}

\thanks{Mich\`{e}le Wigger is with the Communications and Electrical Department, Telecom ParisTech, (e-mail: michele.wigger@telecom-paristech.fr). M.~Wigger is partly supported by the city of Paris under the programme ``Emergences.''}


\thanks{Some of the material in this paper was presented at the IEEE Information Theory Workshop (ITW), Lausanne, Switzerland, September, 2012.}
}

\maketitle

\begin{abstract}
A computable expression for the rate-distortion (RD) function proposed by Heegard and Berger has eluded information theory for nearly three decades. Heegard and Berger's single-letter achievability bound is well known to be optimal for \emph{physically degraded} side information; however, it is not known whether the bound is optimal for arbitrarily correlated side information (general discrete memoryless sources). In this paper, we consider a new setup in which the side information at one receiver is \emph{conditionally less noisy} than the side information at the other. The new setup includes degraded side information as a special case, and it is motivated by the literature on degraded and less noisy broadcast channels. Our key contribution is a converse proving the optimality of Heegard and Berger's achievability bound in a new setting. The converse rests upon a certain \emph{single-letterization} lemma, which we prove using an information theoretic telescoping identity {recently presented by Kramer}. We also generalise the above ideas to two different successive-refinement problems. 
\end{abstract}

\clearpage
\section{Introduction}\label{Sec:Intro}

\IEEEPARstart{W}{yner} and Ziv's seminal 1976 paper~\cite{Wyner-Jan-1976-A} extended rate-distortion (RD) theory to include side information at the receiver. Nearly a decade later, Heegard and Berger~\cite{Heegard-Nov-1985-A} extended the problem setup of~\cite{Wyner-Jan-1976-A} to include multiple receivers with side information: an example of which, and the principal subject of this paper, is shown in Fig.~\ref{Fig:SourceCoding}. The RD function of this problem, however, has eluded complete characterisation in the sense that matching (computable~\cite[p.~259]{Csiszar-1981-B}) achievability and converse bounds have yet to be obtained for general discrete memoryless sources\footnote{Matsuta and Uyematsu~\cite{Matsuta-Jul-2012-C} recently presented matching achievability and converse bounds for Heegard and Berger's RD function using an information-spectrum approach; these bounds, however, are not computable.}.

The best single-letter achievability bound for two receivers is due to Heegard and Berger~\cite[Thm.~2]{Heegard-Nov-1985-A}, and the best bound for three or more receivers is due to Timo, Chan and Grant~\cite[Thm.~2]{Timo-Aug-2011-A}. Both bounds hold for arbitrary discrete memoryless sources under average per-letter distortion constraints. Matching converses have been obtained for some special cases, with each proof being constructed on a case by case basis, e.g.,~\cite{Sgarro-Mar-1977-A,Heegard-Nov-1985-A,Timo-Nov-2010-A,Watanabe-Aug-2011-C}. A special case of note is when the side information is \emph{physically degraded} in the sense that the side information at one receiver is a noisy version of the side information at the other. Heegard and Berger exploited this degraded stochastic structure in their converse~\cite[pp.~733-734]{Heegard-Nov-1985-A} to prove the optimality of their achievability bound.

In this paper, we consider a new setup in which the side information at one receiver is \emph{conditionally less noisy} than the side information at the other. The setup includes physically degraded side information as a special case, and it is motivated by similar, but apparently unrelated, literature on degraded and less noisy broadcast channels~\cite{Korner-1975-C,El-Gamal-2011-B}. Our key contribution is a new converse that proves the optimality of Heegard and Berger's achievability bound in a new setting (conditionally less noisy sources with a deterministic-distortion function at one receiver). The converse rests upon a certain \emph{single-letterization} lemma, which we prove using an information-theoretic telescoping identity recently presented by Kramer in~\cite[Sec.~G]{Kramer-Dec-2011-A}. 

Elements of Heegard-Berger's problem have appeared in many guises throughout the information theory literature. Special cases of the problem include the almost lossless setup of~\cite{Sgarro-Mar-1977-A}, the complementary side information setup of~\cite{Kimura-Apr-2008-A,Timo-Nov-2010-A}, and the product side information setup of~\cite{Watanabe-Aug-2011-C}. Generalisations of the problem include the Wyner-Ziv successive-refinement work of~\cite{Steinberg-Aug-2004-A,Tian-Aug-2007-A,Tian-Dec-2008-A} and the joint source-channel coding setup of~\cite{Tuncel-Apr-2006-A,Nayak-Apr-2010-A1,Gao-Sep-2011-A}. Other variations of the problem have been investigate with causal side information~\cite{Maor-Jan-2008-A,Timo-Jun-2009-C} and common reconstructions~\cite{Ahmadi-Jul-2012-C}. The converse methods presented in this paper may be applicable to these and other problems, particularly to those with existing results on physically degraded side information. Indeed, to conclude the paper, we apply our converse methods to obtain new results for two successive-refinement problems with side information. 

\emph{Paper Outline:} The remainder of the paper is divided into three sections: Section~\ref{Sec:Single-Letter-Lemmas} presents the single-letterization lemma that will be key to our main results (converses); Section~\ref{Sec:Heegard-Berger} presents a new converse for Heegard and Berger's RD problem shown in Fig.~\ref{Fig:SourceCoding}; and Section~\ref{Sec:Successive-Refinement} presents new converses for two successive-refinement problems with side information (physically degraded side information~\cite{Steinberg-Aug-2004-A,Tian-Aug-2007-A} and scalable side information~\cite{Tian-Dec-2008-A}).


\emph{Notation:} All random variables in this paper are discrete and finite and denoted by uppercase letters, e.g., $X$. The alphabet of a random variable is written in matching calligraphic font, e.g. $\set{X}$ is the alphabet of $X$. The $n$-fold Cartesian product of an alphabet is denoted by boldface font, e.g. $\bm{\set{X}}$ is the $n$-fold product of $\set{X}$. If a random vector $(X,Y,Z)$ forms a Markov chain in the same order ($X$ is conditionally independent of $Z$ given $Y$), then we write $X \markov Y \markov Z$. The symbol $\oplus$ denotes modulo-two addition. 


\section{A Lemma}\label{Sec:Single-Letter-Lemmas}

This section concerns a single-letterization (or, entropy-characterisation) problem: express the difference of two $n$-letter conditional mutual informations with a single-letter expression. The lemma in this section is used to prove our converse results.

Consider a tuple of random variables $(R,S_1,S_2,T,L)$ with an arbitrary joint distribution. Let
\begin{equation}
(\bm{R},\bm{S_1},\bm{S_2},\bm{T},\bm{L}) 
\triangleq 
(R_1,S_{1,1},S_{2,1},T_1,L_1),
(R_2,S_{1,2},S_{2,2},T_2,L_2),
\ldots,
(R_n,S_{1,n},S_{2,n},T_n,L_n)
\end{equation}
denote an $n$-tuple of $n$ independent and identically distributed (i.i.d.) tuples of $(R,S_1,S_2,T,L)$. Further, suppose that $J$ is jointly distributed with the $n$-tuple $(\bm{R},\bm{S_1},\bm{S_2},\bm{T},\bm{L})$ and 
\begin{equation}
J\ \markov (\bm{R},\bm{L})\ \markov (\bm{S_1},\bm{S_2}, \bm{T})
\end{equation}
forms a Markov chain. Consider the following difference of $n$-letter conditional mutual informations:
\begin{equation}
I(J; \bm{S_2}|\bm{L}) - I(J; \bm{S_1}|\bm{L}).
\end{equation}
We wish to know whether this difference can be expressed in a \emph{single-letter} form in the sense of Csisz\'{a}r and K\"{o}rner~\cite[p.~259]{Csiszar-1981-B}. The next lemma answers this question in the affirmative.

\begin{lemma}\label{Lem:Conditional-Mathis-Lemma}
Let $(J,\bm{R}, \bm{S_1}, \bm{S_2},\bm{T},\bm{L})$ be defined as above. There exists an auxiliary random variable $W$, jointly distributed with $(R,S_1,S_2,T,L)$ and with alphabet $\set{W}$,  such that
\begin{equation}\label{Eqn:Mathis-Conditional-Card}
|\set{W}| \leq |\set{R}||\set{L}|,
\end{equation}
\begin{equation}\label{Eqn:Mathis-Conditional-MI-Equality}
I(J; \bm{S_2}|\bm{L}) - I(J; \bm{S_1}|\bm{L}) = n\big( I(W;S_2|L) - I(W;S_1|L) \big)
\end{equation}
and 
\begin{equation}
W\ \markov (R,L)\ \markov (S_1,S_2, T) \label{Eqn:MC0}
\end{equation}
forms a Markov chain. If, in addition, $L$ is a function of $R$, then the chain in~\eqref{Eqn:MC0} can be replaced by 
\begin{equation}
W \markov R \markov (S_1,S_2,T)
\end{equation}
and the cardinality bound in~\eqref{Eqn:Mathis-Conditional-Card} can be tightened to
\begin{equation}
|\set{W}| \leq |\set{R}|.
\end{equation}
\end{lemma}

The proof of Lemma~\ref{Lem:Conditional-Mathis-Lemma}, which is given in Appendix~\ref{App:Proof-Mathis-Lemmas}, makes use of an information-theoretic telescoping identity recently presented by Kramer in~\cite[Sec.~G]{Kramer-Dec-2011-A}. 


\section{The Heegard-Berger Problem}\label{Sec:Heegard-Berger}

This section is devoted to Heegard and Berger's RD problem shown in Fig.~\ref{Fig:SourceCoding}. Finding a computable expression for this RD function is a classic, longstanding, open problem in information theory. The section is arranged as follows: we recall the RD function's operational definition in Section~\ref{Sec:HB-ProblemStatement}, we review Heegard and Berger's existing results for degraded side information in Section~\ref{Sec:HB-Existing-Results}, and we state our new results in Section~\ref{Sec:HB-CLN}. 


\subsection{Operational Definition of the RD Function}\label{Sec:HB-ProblemStatement}

Consider a tuple of random variables $(X,Y_1,Y_2)$ with an arbitrary joint distribution on $\set{X} \times \set{Y}_1 \times \set{Y}_2$. Let $(\bm{X},\bm{Y_1},\bm{Y_2})$ denote a string of $n$-i.i.d. random vectors $(X,Y_1,Y_2)$, and let $\bm{\set{X}}$, $\bm{\set{Y}_1}$, $\bm{\set{Y}_2}$ denote the $n$-fold Cartesian products of $\set{X}$, $\set{Y}_1$ and $\set{Y}_2$ respectively. Consider the setup of Fig.~\ref{Fig:SourceCoding}: the Transmitter observes $\bm{X}$, Receiver~1 observes $\bm{Y_1}$ and Receiver~2 observes~$\bm{Y_2}$. The string $\bm{X}$ is to be compressed by the Transmitter and reconstructed by both receivers using a block code. The RD function is the smallest rate at which $\bm{X}$ can be compressed, while allowing the receivers to reconstruct $\bm{X}$ to within specified average distortions. 

An $n$-block code for the setup shown in Fig.~\ref{Fig:SourceCoding} consists of three (possibly stochastic) maps. We denote these maps by
\begin{equation}
f : \bm{\set{X}} \longrightarrow \set{M}
\end{equation}
and
\begin{align}
g_j &: \set{M} \times \bm{\set{Y}_j} \longrightarrow \bm{\hat{\set{X}}_j},\quad j = 1,2,
\end{align}
where $\set{M}$ is a finite index set with cardinality $|\set{M}|$ depending on $n$,  $\hat{\set{X}}_j$ is the reconstruction alphabet of Receiver~$j$ and $\bm{\hat{\set{X}}_j}$ its $n$-fold Cartesian product. The Transmitter sends $M \triangleq f(\bm{X})$ and Receiver~$j$ reconstructs $\bm{\hat{X}_j}$ $\triangleq$ $g_j(M,\bm{Y_j})$.

\begin{figure}
\begin{center}
\includegraphics[width=0.45\columnwidth]{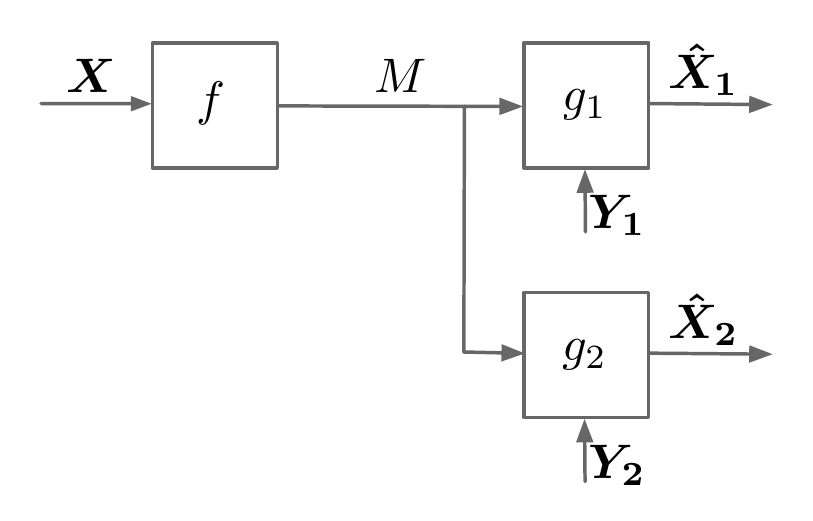}
\end{center}
\caption{Rate distortion with side information at two receivers.}
\label{Fig:SourceCoding}
\end{figure}

Let
\begin{align}
\delta_j &: \set{X} \times \hat{\set{X}}_j \longrightarrow [0,\infty), \quad j = 1,2,
\end{align}
be bounded per-letter distortion functions. For simplicity, and without loss of generality, we assume that $\delta_1$ and $\delta_2$ are \emph{normal}~\cite[p.~185]{Yeung-2008-B}; that is, for each $x$ in $\set{X}_j$ there exists some $\hat{x}$ in $\hat{\set{X}}_j$ such that $\delta_j(x,\hat{x}) = 0$.

\begin{definition}\label{Def:Achievable-Rate}
A rate $R$ is said to be $(D_1,D_2)$-\emph{achievable} if for each $\epsilon > 0$ there exists an $n$-block code $(f,g_1,g_2)$, for some sufficiently large blocklength $n$, satisfying 
\begin{align}\label{Eqn:Achievable-Rate}
R + \epsilon &\geq \frac{1}{n} \log |\set{M}|  
\end{align}
and
\begin{align}
D_j + \epsilon &\geq \mathbb{E}  \frac{1}{n} \sum_{i=1}^n \delta_j(X_i,\hat{X}_{j,i}), \quad j = 1,2.
\end{align}
\end{definition}

\begin{definition}[RD Function]\label{Def:RD-Function}
\begin{equation}
R(D_1,D_2) \triangleq \min \big\{ R > 0 : R  \text{ is $(D_1,D_2)$-achievable} \big\}, \quad D_1 \geq 0,\ D_2 \geq 0.
\end{equation}
\end{definition}


\subsection{Existing Results}\label{Sec:HB-Existing-Results}

Computable \emph{single-letter}~\cite{Csiszar-1981-B} expressions for the RD function have been found in some special cases, see~\cite{Heegard-Nov-1985-A,Timo-Nov-2010-A,Watanabe-Aug-2011-C}. The achievability proofs of these cases all follow from a result by Heegard and Berger~\cite{Heegard-Nov-1985-A}, which we review in the next lemma. The converses, in contrast, are derived on a case-by-case basis. 

\begin{lemma}[Achievability]\label{Lem:HB}
The RD function is bound from above by~\cite[Thm.~2]{Heegard-Nov-1985-A}
\begin{equation}
\begin{split}
R(D_1,D_2) \leq \min_{(A,B,C)} \Big\{ \max\big\{I(X;C|Y_1),I(X;C|Y_2)\big\} + I(X;A|C,Y_1) + I(X;B|C,Y_2) \Big\},\label{Eqn:RD1D2RHS}
\end{split}
\end{equation}
where minimisation is taken over all auxiliary random variables $(A,B,C)$, jointly distributed with the source $(X,Y_1,$ $Y_2)$, such that the following is true:
\begin{itemize}
\item[(i)] the auxiliary random variables are conditionally independent of the side information given $X$,
\begin{equation}\label{Eqn:HB-Achievability-Markov-Chain}
(A,B,C)\ \markov X\ \markov (Y_1,Y_2);
\end{equation}
\item[(ii)] the cardinalities of the alphabets of $C$, $A$ and $B$ are respectively bound by 
\begin{subequations}\label{Eqn:HB-Achievability-Card-Bounds}
\begin{align}
|\hspace{0.4mm}\set{C}\hspace{0.4mm}| &\leq |\set{X}| + 3 \\
|\set{A}| &\leq |\set{C}| |\set{X}| + 1\\
|\hspace{0.2mm}\set{B}\hspace{0.3mm}| &\leq |\set{C}| |\set{X}| + 1
\end{align}
(these cardinality bounds are new, see Appendix~\ref{App:HBcardinality} for our proof);
\end{subequations}
\item[(iii)] there exist deterministic maps
\begin{subequations}\label{Eqn:Reconstruction}
\begin{align}
\phi_1 &: \set{A} \times \set{C} \times \set{Y}_1 \longrightarrow \hat{\set{X}}_1\\
\phi_2 &: \set{B} \times \set{C} \times \set{Y}_2 \longrightarrow \hat{\set{X}}_2
\end{align}
\end{subequations}
with 
\begin{subequations}\label{Eqn:distortions}
\begin{align}
D_1 &\geq \mathbb{E}\ \delta_1\big(X,\phi_1(A,C,Y_1)\big)\\
D_2 &\geq \mathbb{E}\ \delta_2\big(X,\phi_2(B,C,Y_2)\big).
\end{align}
\end{subequations}
\end{itemize}
\end{lemma}

The next definition and theorem review a special case for which the upper bound of Lemma~\ref{Lem:HB} is tight.

\begin{definition}\label{Def:DegradedSI}
The side information is said to be \emph{physically degraded} if  
\begin{equation}\label{Eqn:Physically-Degraded}
X \markov Y_2 \markov Y_1. 
\end{equation}
\end{definition}

\begin{theorem}\label{Thm:DegradedSI}
If the side information is physically degraded, then~\cite[Thm.~3]{Heegard-Nov-1985-A}
\begin{equation}
R(D_1,D_2) = \min_{(B,C)}\Big\{ I(X;C|Y_1) +  I(X; B |C ,Y_2) \Big\},
\end{equation}
where the minimisation is taken over all auxiliary $(B,C)$, jointly distributed with $(X,Y_1,Y_2)$, such that 
\begin{enumerate}
\item[(i)] the auxiliary random variables are conditionally independent of the side information given $X$,
\begin{equation}
(B,C) \markov X \markov (Y_1,Y_2);
\end{equation}
\item[(ii)] there exist deterministic maps
\begin{align}
\phi_1 &: \set{C} \times \set{Y}_1 \longrightarrow \hat{\set{X}}_1\\
\phi_2 &: \set{B} \times \set{C} \times \set{Y}_2 \longrightarrow \hat{\set{X}}_2
\end{align}
with 
\begin{align}
D_1 &\geq \mathbb{E}\ \delta_1\big(X,\phi_1(C,Y_1)\big)  \\
D_2 &\geq \mathbb{E}\ \delta_2\big(X,\phi_2(B,C,Y_2)\big).
\end{align}
\end{enumerate}
\end{theorem}

The Markov chain in~\eqref{Eqn:Physically-Degraded}, which defines physically degraded side information, enables a crucial step in Heegard and Berger's converse of Theorem~\ref{Thm:DegradedSI}, see~\cite[pp.~733-734]{Heegard-Nov-1985-A}. The goal of the next section is to broaden the scope of Theorem~\ref{Thm:DegradedSI} by replacing the Markov chain~\eqref{Eqn:Physically-Degraded} with a more general condition. Our main results, however, will fall slightly short of this goal: we will need to restrict attention to the setting where Receiver~1 requires an almost lossless copy of a function of $X$. More specifically, we will require that $D_1 = 0$ and $\delta_1$ is deterministic in the following sense. 

\begin{definition}\label{Def:DeterministicDistortion}
$\delta_1$ is said to be \emph{deterministic}~\cite{El-Gamal-Nov-1982-A,Tian-Dec-2008-A}  if there is an alphabet $\tilde{\set{X}}$ with $\hat{\set{X}}_1 = \tilde{\set{X}}$ and a deterministic map 
\begin{equation}
\psi : \set{X} \longrightarrow \tilde{\set{X}} 
\end{equation}
such that 
\begin{equation}
\delta_1(x,\hat{x}) \triangleq \left\{
\begin{array}{ll}
0 & \text{ if } \hat{x} = \psi(x)\\
1 & \text{ otherwise}.
\end{array}
\right.
\end{equation}
\end{definition}

For later discussions, we need to specialise Theorem~\ref{Thm:DegradedSI} to deterministic $\delta_1$. Let 
\begin{equation}
\tilde{X} \triangleq \psi(X).
\end{equation}
Define
\begin{equation}\label{Eqn:Wyner-Ziv-1}
S(D_2) \triangleq  \min_B\ I(X;B|\tilde{X},Y_2),\quad D_2 \geq 0,
\end{equation}
where the minimisation is taken over all auxiliary $B$, jointly distributed with $(X,Y_1,Y_2)$, such that \begin{enumerate}
\item[(i)] the auxiliary random variable $B$ is conditionally independent of the side information $(Y_1,Y_2)$ given $X$,
\begin{equation}
B\ \markov X \ \markov (Y_1,Y_2);
\end{equation}
\item[(ii)] the cardinality of the alphabet of $B$ is bound by
\begin{equation}
|\set{B}| \leq |\set{X}| + 1;
\end{equation}
\item[(iii)] there exists deterministic 
\begin{equation}
\phi_2 : \set{B} \times \tilde{\set{X}} \times \set{Y}_2 \longrightarrow \hat{\set{X}}_2
\end{equation}
with
\begin{equation}
D_2 \geq \mathbb{E}\ \delta_2\big(X,\phi_2(B,\tilde{X},Y_2)\big).
\end{equation}
\end{enumerate}
The function $S(D_2)$ is non-increasing, convex and continuous in $D_2$, see~\cite[Thm.~A2]{Wyner-Jan-1976-A}. The next corollary is proved in Appendix~\ref{App:DeterministicDegradedSI}.

\begin{corollary}\label{Cor:DeterministicDegradedSI}
If the side information is physically degraded and $\delta_1$ is deterministic, then 
\begin{equation}
R(0,D_2) = H(\tilde{X}|Y_1) + S(D_2).
\end{equation}
\end{corollary}

It will be useful to further specialise Corollary~\ref{Cor:DeterministicDegradedSI} to the following two-source with component Hamming distortion functions. This specialisation is central to our understanding of how Corollary~\ref{Cor:DeterministicDegradedSI} can be generalised. 

\begin{definition}
We say that $(X,Y_1,Y_2)$ is a \emph{two-source} if 
\begin{equation}
\set{X} \triangleq \set{X}_1 \times \set{X}_2\quad \text{and}\quad X \triangleq (X_1,X_2),
\end{equation}
where $\set{X}_1$ and $\set{X}_2$ are finite alphabets. In addition, we say that $\delta_1$ and $\delta_2$ are \emph{component Hamming distortion functions} if
\begin{equation}
\hat{\set{X}}_j = \set{X}_j
\end{equation}
and
\begin{equation}
\delta_j(x,\hat{x}) = \left\{
\begin{array}{ll}
0 & \text{ if } \hat{x} = x\\
1 & \text{ otherwise}
\end{array}
\right.
\end{equation}
for $j = 1,2$.
\end{definition}

\begin{corollary}\label{Cor:LosslessDegradedSI} 
Consider a two-source $(X_1,X_2,Y_1,Y_2)$ with component Hamming distortion functions. If the side information is physically degraded, i.e., 
\begin{equation}\label{Eqn:Two-Source-DegradedSI}
(X_1,X_2) \markov Y_2 \markov Y_1,
\end{equation}
then~\cite{Heegard-Nov-1985-A,Timo-Aug-2011-A} 
\begin{equation}\label{Eqn:Cor:LosslessDegradedSI-Rate}
R(0,0) = H(X_1|Y_1) + H(X_2|X_1,Y_2).
\end{equation}
\end{corollary}

The last corollary can be directly proved in a simple way that nicely adds motivation to the possibility of a more general converse.

\emph{Proof Outline (Converse):}  If $R$ is achievable, then for each $\epsilon > 0$ and sufficiently large $n$ there exists an $n$-block code $(f,g_1,g_2)$ for which the following is true: 
\begin{align}
%
%
R + \epsilon &\geq \frac{1}{n} \log |\set{M}|\\
%
%
&\geq \frac{1}{n} H(M)\\
%
%
&\geq \frac{1}{n} I(\bm{X_1},\bm{X_2},\bm{Y_1},\bm{Y_2};M)\\
%
%
&= \frac{1}{n} \Big( I(\bm{X_1},\bm{Y_1};M) 
+ I(\bm{X_2},\bm{Y_2};M|\bm{X_1},\bm{Y_1}) \Big) \\
%
%
&\geq \frac{1}{n} \Big( I(\bm{X_1};M|\bm{Y_1}) 
+ I(\bm{X_2};M|\bm{X_1},\bm{Y_1},\bm{Y_2}) \Big) \\
%
%
&\step{a}{\geq} \frac{1}{n} \Big( H(\bm{X_1} | \bm{Y_1})  + H(\bm{X_2}|\bm{X_1},\bm{Y_1},\bm{Y_2}) 
 - n\varepsilon(n,\epsilon) \Big)\\
 %
 %
&\step{b}{=} H(X_1 | Y_1)  + H(X_2|X_1,Y_1,Y_2) - \varepsilon(n,\epsilon)\\
 \label{Eqn:Converse:Two-Source-Hamming-Degraded}
&\step{c}{=} H(X_1 | Y_1)  + H(X_2 | X_1, Y_2) - \varepsilon(n,\epsilon).
\end{align}
The justification for steps (a), (b) and (c) is as follows.
\begin{enumerate}
\item[(a)] $\bm{\hat{X}_1}$ and $\bm{\hat{X}_2}$ are determined by $(M,\bm{Y_1})$ and $(M,\bm{Y_2})$ respectively, so (a) follows by Fano's inequality~\cite[Sec.~2.2]{El-Gamal-2011-B}. Here the function $\varepsilon(n,\epsilon)$ can be chosen so that $\varepsilon(n,\epsilon)   \rightarrow 0$ as $\epsilon \rightarrow 0$.
\item[(b)] $(\bm{X_1},\bm{X_2},\bm{Y_1},\bm{Y_2})$ is i.i.d.
\item[(c)] The side information is physically degraded and consequently $X_2\ \markov (X_1,Y_2)\ \markov Y_1$.
\end{enumerate}

\emph{Proof outline (achievability):} Suppose that we use the Slepian-Wolf / Cover random-binning argument to send $\bm{X_1}$ losslessly to Receiver~1 at rate $R'$ close to $H(X_1|Y_1)$. The side information is physically degraded, so we have
\begin{equation}\label{Eqn:Slepian-Wolf-Motivation}
R' \geq H(X_1|Y_1) \geq H(X_1|Y_2).
\end{equation}

A close inspection of the random binning proof, e.g.~\cite{El-Gamal-2011-B}, reveals that~\eqref{Eqn:Slepian-Wolf-Motivation} also suffices for Receiver~2 to reliably decode~$\bm{X_1}$. Now, assuming $\bm{X_1}$ is successfully decoded by Receiver~2, we can send $\bm{X}_2$ to Receiver~2 at a rate $R''$ close to $H(X_2|X_1,Y_2)$ using $(\bm{X_1},\bm{Y_2})$ as side information. The total rate $R = R'+R''$ is close to $H(X_1|Y_1) + H(X_2|X_1,Y_2)$. \hfill $\blacksquare$

\medskip

We notice that the Markov chain in~\eqref{Eqn:Two-Source-DegradedSI} is equivalent to 
\begin{subequations}\label{Eqn:Cor:LosslessDegradedSI:MarkovChains}
\begin{equation}\label{Eqn:Cor:LosslessDegradedSI:MarkovChains:A}
X_1\ \markov Y_2\ \markov Y_1
\end{equation}
and
\begin{equation}\label{Eqn:Cor:LosslessDegradedSI:MarkovChains:B}
X_2\ \markov (X_1,Y_2)\ \markov Y_1.
\end{equation}
\end{subequations}
The chain~\eqref{Eqn:Cor:LosslessDegradedSI:MarkovChains:A} is a sufficient, but not necessary, condition for the inequalities in~\eqref{Eqn:Slepian-Wolf-Motivation} and hence the above achievability argument. In contrast, the chain~\eqref{Eqn:Cor:LosslessDegradedSI:MarkovChains:B} is essential for equality (c) in~\eqref{Eqn:Converse:Two-Source-Hamming-Degraded} and hence the converse argument. The generality of the achievability argument juxtaposed against the more restrictive converse argument suggests that~\eqref{Eqn:Cor:LosslessDegradedSI-Rate} might hold for a broader class of two-sources. We show that this is indeed the case in the next subsection; specifically, we will see that~\eqref{Eqn:Cor:LosslessDegradedSI-Rate} still holds when the Markov chain~\eqref{Eqn:Cor:LosslessDegradedSI:MarkovChains:A} is replaced by $H(X_1|Y_1) \geq H(X_1|Y_2)$ and the chain~\eqref{Eqn:Cor:LosslessDegradedSI:MarkovChains:B} is replaced by a more general ``conditionally less noisy'' condition.

\begin{remark}
\begin{enumerate}
\item[]
\item[(i)] $R(D_1,D_2)$ depends on the joint distribution of $(X,Y_1,Y_2)$ only via the marginal distributions of $(X,Y_1)$ and $(X,Y_2)$.  
\item[(ii)] The side information is said to be \emph{stochastically degraded} if the joint distribution of $(X,Y_1,Y_2)$ is such that there exists some physically degraded side information $(X',Y_1',Y_2')$ with marginals $(X',Y_1')$ and $(X',Y_2')$ matching those of $(X,Y_1)$ and $(X,Y_2)$. By Remark~1~(i), Theorem~\ref{Thm:DegradedSI} and Corollaries~\ref{Cor:DeterministicDegradedSI} and~\ref{Cor:LosslessDegradedSI} also hold for stochastically degraded side information. 
\item[(iii)] The function $S(D_2)$, which is defined in~\eqref{Eqn:Wyner-Ziv-1}, is the Wyner-Ziv RD function~\cite[Eqn.~(15)]{Wyner-Jan-1976-A} for a source $\bm{X}$ with side information $(\bm{\tilde{X}},\bm{Y_2})$. 
\item[(iv)] The asserted upper bound for $R(D_1,D_2)$ in~\cite[Thm.~2]{Heegard-Nov-1985-A} is incorrect for the case of three or more receivers~\cite{Timo-Aug-2011-A}.
\end{enumerate}
\end{remark}

\subsection{New Results}\label{Sec:HB-CLN}

Suppose that $L$ is an auxiliary random variable that is jointly distributed with the source $(X,Y_1,Y_2)$.

\begin{definition}\label{Def:ConditionallyLessNoisy}
We say that $Y_2$ is \emph{conditionally less noisy than $Y_1$ given} $L$, abbreviated as $\cln{Y_2}{Y_1}{L}$, if 
\begin{equation}
I(W;Y_2|L) \geq I(W;Y_1|L)
\end{equation}
holds for every auxiliary $W$, jointly distributed with $(X,Y_1,Y_2,L)$, for which
\begin{equation}\label{Eqn:MClessnoisy}
W \markov (X,L) \markov (Y_1,Y_2).
\end{equation}
\end{definition}

The next lemma and example collectively show that Definition~\ref{Def:ConditionallyLessNoisy} is broader than Definition~\ref{Def:DegradedSI}. The lemma is proved in Appendix~\ref{App:LessNoisySpecialCases}.

\clearpage

\begin{lemma}\label{Lem:LessNoisySpecialCases}\ 
\begin{enumerate}
\item[(i)] If the side information $(X,Y_1,Y_2)$ is physically degraded and the auxiliary random variable $L$ satisfies the Markov chain
\begin{equation}\label{Eqn:LessNoisySpecialCases-MCinLemma1} 
L\ \markov X\ \markov (Y_1,Y_2),
\end{equation}
then $\cln{Y_2}{Y_1}{L}$.
\item[(ii)] If a two-source $(X_1,X_2,Y_1,Y_2)$ satisfies 
\begin{equation}\label{Eqn:LessNoisySpecialCases-MCinLemma2} 
X_2\ \markov X_1\ \markov Y_1
\end{equation}
and $L = X_1$, then $\cln{Y_2}{Y_1}{X_1}$.
\end{enumerate}
\end{lemma}

The next example describes a two-source, where the side information is \emph{not} degraded, but~\eqref{Eqn:LessNoisySpecialCases-MCinLemma2} holds and therefore $\cln{Y_2}{Y_1}{X_1}$.

\begin{example}\label{Ex:CLNSource}
Let $X_2$, $Y_2$, and $Z$ be independent Bernoulli random variables with 
\begin{align}
\mathbb{P}[X_2 = 0] &= 1 - \mathbb{P}[X_2 = 1] = p,\quad p \in (0,1/2),\\
\mathbb{P}[Y_2 = 0] &= 1 - \mathbb{P}[Y_2 = 1] = q,\quad q \in (0,1/2),\\
{\mathbb{P}[Z = 0] }&{= 1 - \mathbb{P}[Z = 1] = r,\quad r \in (0,1/2).}
\end{align}
Let
\begin{equation}
X_1 = X_2 \oplus Y_2 
\end{equation}
and 
\begin{equation}
Y_1=X_1\oplus Z. 
\end{equation}
We have 
\begin{equation}
X_2\ \markov X_1\ \markov Y_1,
\end{equation}
so assertion~(ii) of Lemma~\ref{Lem:LessNoisySpecialCases} implies $\cln{Y_2}{Y_1}{X_2}$. In contrast, $(X_1,X_2)$ is not conditionally independent of $Y_1$ given $Y_2$ and, therefore, the side information is not physically degraded. 
\end{example}

The next lemma gives a lower bound for the RD function. Its proof uses the single-letterization Lemma~\ref{Lem:Conditional-Mathis-Lemma} and is the subject of Appendix~\ref{Sec:ConverseProof}. Our main result in this section, Theorem~\ref{Thm:LossyCoding}, follows directly thereafter.

\begin{lemma}[Converse]\label{Lem:Converse}
If $\delta_1$ is deterministic, then the following is true. 
\begin{enumerate}
\item[(i)] For arbitrarily distributed $(X,Y_1,Y_2)$, we have 
\begin{equation}
R(0,D_2) \geq  H(\tilde{X}|Y_1) + S(D_2) + \min \big\{ I(W;Y_2|\tilde{X}) - I(W;Y_1|\tilde{X}) \big\},
\end{equation}
where the minimisation is taken over all auxiliary $W$, jointly distributed with $(X,Y_1,Y_2)$, such that 
\begin{equation}
W \markov X \markov (Y_1,Y_2),
\end{equation}
and 
\begin{equation}
|\set{W}| \leq |\set{X}|.
\end{equation}
\item[(ii)] If $(X,Y_1,Y_2)$ satisfies $\cln{Y_2}{Y_1}{\tilde{X}}$, then 
\begin{equation}
R(0,D_2) \geq H(\tilde{X}|Y_1) + S(D_2).
\end{equation}
\end{enumerate}
\end{lemma}

It is worth highlighting that in the minimisation 
\begin{equation}\label{Eqn:Converse-Min-Arbitrary-Sources}
\min \big\{ I(W;Y_2|\tilde{X}) - I(W;Y_1|\tilde{X})\big\}
\end{equation} 
it is always possible to choose $W$ to be constant and \eqref{Eqn:Converse-Min-Arbitrary-Sources} must therefore be non-positive. Assertion (ii) of the lemma follows immediately from assertion (i) upon invoking Definition~\ref{Def:ConditionallyLessNoisy} with the auxiliary random variable $L = \tilde{X}$.  

The next theorem gives a single-letter expression for $R(D_1,D_2)$ in a new setting, and it is the main result of this section. The theorem is a direct consequence of the achievability of Lemma~\ref{Lem:HB} and the converse of Lemma~\ref{Lem:Converse}~(ii). 

\begin{theorem}\label{Thm:LossyCoding}
If $\delta_1$ is deterministic, 
\begin{equation}\label{Eqn:Thm-LossyCoding-Conditions}
\cln{Y_2}{Y_1}{\tilde{X}} \text{ and } H(\tilde{X}|Y_1) \geq H(\tilde{X}|Y_2),
\end{equation}
then
\begin{equation}\label{Eqn:Thm-LossyCoding-Rate}
R(0,D_2) = H(\tilde{X}|Y_1) + S(D_2).
\end{equation}
\end{theorem}

\begin{proof}
The achievability of~\eqref{Eqn:Thm-LossyCoding-Rate} follows from Lemma~\ref{Lem:HB} where we set $C =\tilde{X}$ and $A =$ constant. 
The converse follows by Lemma~\ref{Lem:Converse}.
\end{proof}

The next corollary generalises Corollary~\ref{Cor:LosslessDegradedSI} to the conditionally less noisy setting.

\begin{corollary}\label{Cor:LossyCoding}
Consider a two-source and component Hamming distortion functions. If
\begin{equation}\label{Eqn:Cor:LossyCoding}
\cln{Y_2}{Y_1}{X_1}\ \text{ and }\ H(X_1|Y_1) \geq H(X_1|Y_2),
\end{equation}
then 
\begin{equation}
R(0,0) = H(X_1|Y_1) + H(X_2|X_1,Y_2).
\end{equation}
\end{corollary}

\begin{proof}
In Theorem~\ref{Thm:LossyCoding}, we have $\tilde{X} = X_1$ and
\begin{equation}
S(0) = H(X_2|X_1,Y_2).
\end{equation}
\end{proof}

\begin{example}\label{Exa:BECBSC}
Let $X_1$ and $Z$ be independent Bernoulli random variables with
\begin{equation}
\mathbb{P}[X_1=0] = \mathbb{P}[X_1=1] = \frac{1}{2}
\end{equation}
and
\begin{equation}
\mathbb{P}[Z=0] = 1- \mathbb{P}[Z=1] = \frac{1}{3}.
\end{equation}
Let 
\begin{equation}
X_2=X_1 \oplus Z.
\end{equation}
Let $Y_2$ and $Y_1$ be the outcomes of passing $X_1$ through a BEC$(2/3)$ and a BSC$(1/4)$ respectively, see Fig.~\ref{Fig:BECBSCExample}. We have $\cln{Y_2}{Y_1}{X_1}$ from condition (ii) of Lemma~\ref{Lem:LessNoisySpecialCases}. 
Moreover, 
\begin{equation}
H(X_1|Y_2)= \frac{2}{3}
\end{equation}
is smaller than 
\begin{equation}
H(X_1|Y_1)=H_b(1/4) \approx 0.8113,
\end{equation}
where 
\begin{equation}
H_\text{b}(\alpha) \triangleq - \alpha \log_2 \alpha - (1-\alpha) \log_2 (1-\alpha)
\end{equation}
is the binary entropy function; therefore, we may apply Corollary~\ref{Cor:LossyCoding} to get
\begin{equation}
R(0,0) = H_\text{b}(1/4)+ H_\text{b}(1/3).
\end{equation}
We notice that since $2/3>2/4$  the side information $Y_2$ and $Y_1$ is not physically or stochastically degraded with respect to $X_1$~\cite[p.~121]{El-Gamal-2011-B}, \cite{Nair-Sep-2010-A}, and hence with respect to $X=(X_1,X_2)$.
\end{example}

\medskip

\begin{remark}
\begin{enumerate}
\item[]
\item[(i)] Theorem~\ref{Thm:LossyCoding} includes Corollary~\ref{Cor:DeterministicDegradedSI} for physically degraded side information as a special case, since  
\begin{equation}
X\ \markov Y_2\ \markov Y_1
\end{equation} 
and
\begin{equation}
\tilde{X}\ \markov X\ \markov (Y_1,Y_2)
\end{equation}
implies~\eqref{Eqn:Thm-LossyCoding-Conditions} by Lemma~\ref{Lem:LessNoisySpecialCases}~(i) and the data processing lemma.
\item[(ii)] It appears that our approach to proving Lemma~\ref{Lem:Converse} (ii) does not readily generalise to an arbitrary distortion function, $\delta_1$. An apparent difficulty follows from the use of a Wyner-Ziv style converse argument to construct the $S(D_2)$ term using $(\bm{\tilde{X}},\bm{Y_1})$ as side information. The argument needs $(\bm{\tilde{X}},\bm{Y_1})$ to be i.i.d. and, if $\delta_1$ is arbitrary, this need not be the case. 
\item[(iii)] Theorem~\ref{Thm:LossyCoding} employs the conditionally less noisy definition for the special case where $L$ is a deterministic function of the source $X$. In this case, we can remove $L$ from the Markov chain in~\eqref{Eqn:MClessnoisy}.
\item[(iv)] If $L=\emptyset$, then Definition~\ref{Def:ConditionallyLessNoisy} reduces to the \emph{less noisy} concept for information-theoretic security for source coding recently introduced by Villard and Piantanida~\cite{Villard-May-2012-A}. Thus, our definition is more broad. In fact, in Example~\ref{Ex:CLNSource} and when the parameter $r$ is sufficiently small (or large) compared to $p$ so that 
\begin{equation}
H(X_1|Y_1)< H(X_2),
\end{equation}
the side information $Y_2$ is \emph{conditionally less noisy} than $Y_1$ given $X_2$, but it is not \emph{less noisy}. To see this, select $W=X_1$, so that 
\begin{equation}
I(W;Y_1)=H(X_1)- H(X_1|Y_1)
\end{equation}
and
\begin{align}
I(W;Y_2) & = H(X_1) - H(X_1|Y_2) \\
& = H(X_1) - H(X_2).
\end{align}
\end{enumerate}
\end{remark}

\begin{figure}
\begin{center}
\includegraphics[width=0.6\columnwidth]{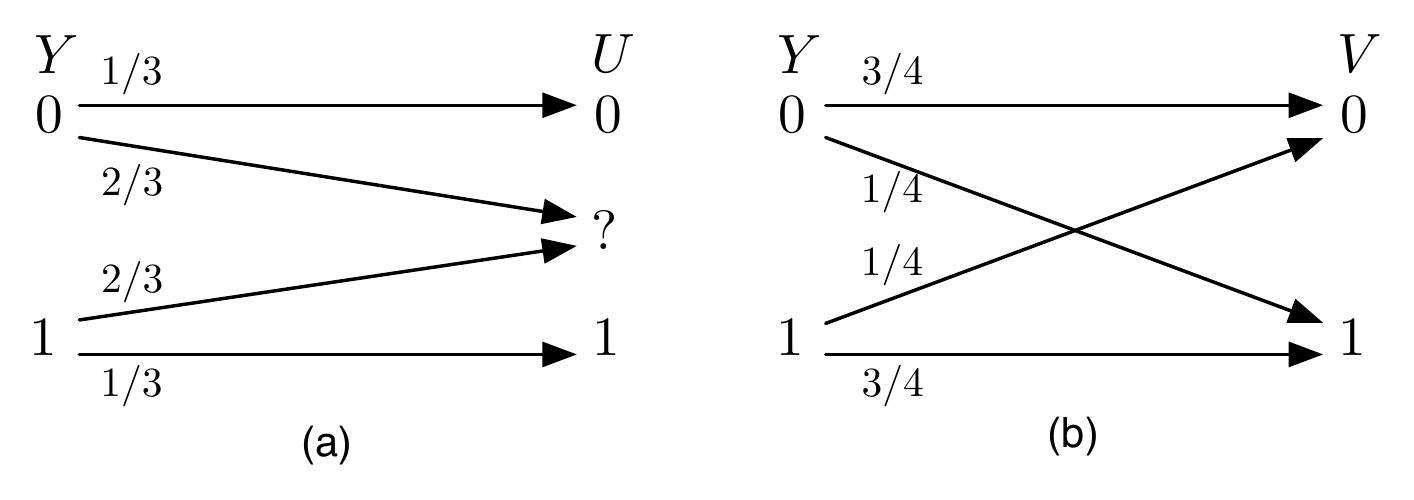}
\caption{Binary channels defining the side information in Example~\ref{Exa:BECBSC}: (a) Binary Erasure Channel (BEC) with erasure probability $2/3$; and (b) Binary Symmetric Channel (BSC) with crossover probability $1/4$. $Y$ and $\tilde{X}$ are Bernoulli $(1/2)$ and $(1/3)$, and $X = Y \oplus \tilde{X}$.} 
\label{Fig:BECBSCExample}
\end{center}
\end{figure}


\clearpage

\section{Successive Refinement with Side Information}\label{Sec:Successive-Refinement}

The method used in Appendix~\ref{Sec:ConverseProof} to prove Lemma~\ref{Lem:Converse} can, with appropriate modification, yield useful converses for various generalisations of Heegard and Berger's RD problem. In this section, we extend the setup of Fig.~\ref{Fig:SourceCoding} to two different successive-refinement problems with receiver side information.

\subsection{Problem Formulation}\label{Sec:Problem formulation}

Consider a tuple of random variables $(X, Y_1,Y_2,Y_3)$ with an arbitrary joint distribution. Let  $(\bm{X}, \bm{Y_1},$ $ \bm{Y_2},\bm{Y_3})$ denote a string of $n$-i.i.d. random vectors $(X, Y_1,Y_2,Y_3)$. A successive-refinement $n$-block code for the setup shown in Fig.~\ref{Fig:SR} consists of four (possibly stochastic) maps 
\begin{equation}
f : \bm{\set{X}} \longrightarrow \set{M}_1 \times \set{M}_2 \times \set{M}_3
\end{equation} 
and 
\begin{align}
g_1 & : \set{M}_1 \times \bm{\set{Y}_1} \longrightarrow \bm{\hat{\set{X}}_1}\\
g_2 & : \set{M}_1 \times \set{M}_2 \times \bm{\set{Y}_2} \longrightarrow \bm{\hat{\set{X}}_2}\\
g_3 & : \set{M}_1 \times \set{M}_2 \times \set{M}_3 \times \bm{\set{Y}_3} \longrightarrow \bm{\hat{\set{X}}_3},
\end{align}
where $\set{M}_1$, $\set{M}_2$ and $\set{M}_3$ are finite sets. The Transmitter sends $(M_1,M_2,M_3) \triangleq f(\bm{X})$ over the noiseless channels, as shown in Fig.~\ref{Fig:SR}. Receiver~1 reconstructs $
\bm{\hat{X}_1} \triangleq g_1(M_1, \bm{Y_1})$, Receiver~2 reconstructs $
\bm{\hat{X}_2} \triangleq g_2(M_1,M_2,\bm{Y_2})$ and Receiver~3 reconstructs $
\bm{\hat{X}_3} \triangleq g_3(M_1,M_2,M_3,\bm{Y_3})$.

\begin{figure}[t]
\begin{center}
\includegraphics[width=0.45\columnwidth]{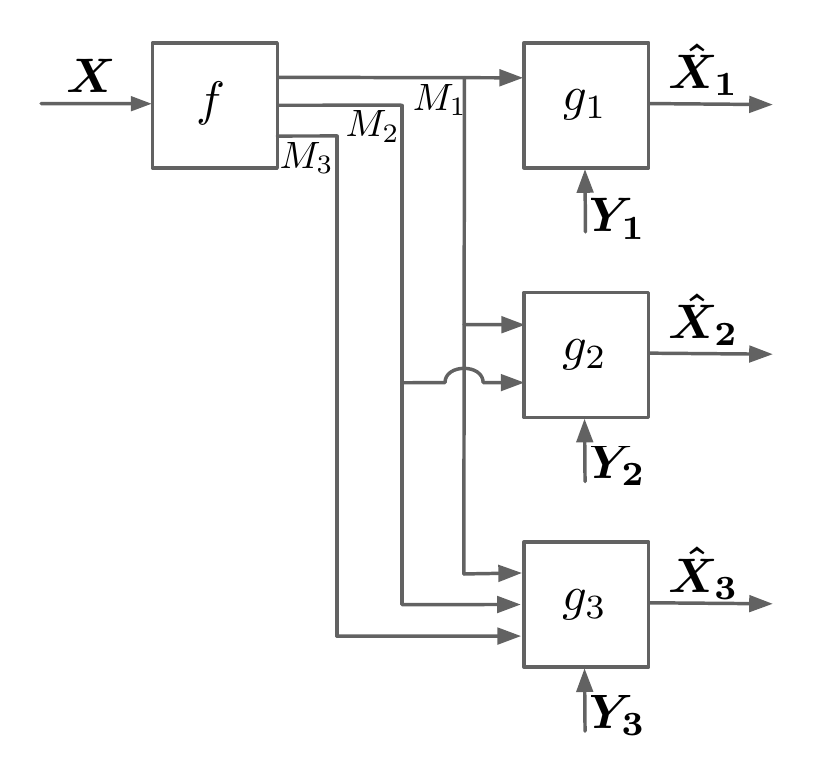}
\caption{Three-stage successive refinement with side information at the receivers.}
\label{Fig:SR}
\end{center}
\end{figure}

\begin{definition}
A rate tuple $(R_1,R_2,R_3)$ is said to be \emph{achievable with distortions} $(D_1,D_2,D_3)$ if for each $\epsilon > 0$ there exists an $n$-block code $(f,g_1,g_2,g_3)$, for some sufficiently large blocklength $n$, satisfying 
\begin{equation}
R_j + \epsilon \geq \frac{1}{n} \log |\set{M}_j|
\end{equation}
and 
\begin{equation}
D_j + \epsilon \geq \mathbb{E} \frac{1}{n} \sum_{i=1}^n \delta_j(X_j,\hat{X}_{j,i})
\end{equation}
for $j = 1,2,3$.
\end{definition}

\begin{definition}[RD Region]
\begin{equation}
\set{R}(D_1,D_2,D_3) \triangleq \big\{(R_1,R_2,R_3) \text{ achievable with distortions }  (D_1,D_2,D_3)\big\},\end{equation}
for $D_1 \geq 0$, $D_2 \geq 0$ and $D_3 \geq 0$.
\end{definition}

\subsection{Three Stages with $Y_3$ better than $Y_2$ better than $Y_1$ (abhinc $X \markov Y_3 \markov Y_2 \markov Y_1$)}
\label{Sec:Successive-Refinement-Deg}

In this subsection, we assume that Receiver~3 obtains the best side information and Receiver~1 the worst.  Tian and Diggavi~\cite{Tian-Aug-2007-A} modelled such a relation with physically degraded side information, i.e., $X \markov Y_3 \markov Y_2 \markov Y_1$, and they derived the corresponding RD region. The goal here is to broaden their result to a conditionally less noisy setup.

We will need the following achievable RD region that holds for arbitrarily distributed side information. The region is  distilled from a more general achievability result in~\cite{Timo-Aug-2011-A}, see Appendix~\ref{App:SR-SSC-Achievability-Gen}. 

Let $\set{R}_\text{in}(D_1,D_2,D_3)$ denote the set of all rate tuples $(R_1,R_2,R_3)$ for which there exist auxiliary random variables $(A_1,A_2,A_3)$, jointly distributed with the source $(X,Y_1,Y_2,Y_3)$, such that the following is true:
\begin{enumerate}
\item[(i)] the auxiliary random variables are conditionally independent of the side information given $X$,
\begin{equation}
(A_1,A_2,A_3)\ \markov X\ \markov (Y_1,Y_2,Y_3);
\end{equation}
\item[(ii)] the cardinalities of the alphabets of $A_1$, $A_2$ and $A_3$ are respectively bound by\footnote{Reference~\cite{Timo-Aug-2011-A} does not provide cardinality constraints. The bounds in~\eqref{Eqn:SR-Achievability-Gen-Card} follow by the standard convex cover method.}
\begin{subequations}\label{Eqn:SR-Achievability-Gen-Card}
\begin{align}
|\set{A}_1| &\leq |\set{X}| + 6\\
|\set{A}_2| &\leq |\set{X}|\ |\set{A}_1| + 4\\
|\set{A}_3| &\leq |\set{X}|\ |\set{A}_1|\ |\set{A}_2| + 1;
\end{align}
\end{subequations}
\item[(iii)] there exist (deterministic) maps for each $j = 1,2,3$
\begin{subequations}
\begin{equation}
\phi_j : \set{A}_j \times \set{Y}_j \longrightarrow \hat{\set{X}}_j
\end{equation}
\end{subequations}
with 
\begin{subequations}
\begin{equation}
D_j \geq \mathbb{E}\ \delta_j\big(X,\phi_j(A_j,Y_j)\big);
\end{equation}
\end{subequations}
\item[(iv)] the rate tuple $(R_1,R_2,R_3)$ satisfies
\begin{subequations}\label{Eqn:Lem:SR-Achievability-Gen-Rates}
\begin{align}
R_1 &\geq I(X;A_1|Y_1),\\
R_1 + R_2 &\geq \max_{j=1,2} I(X;A_1|Y_j) + I(X;A_2|A_1,Y_2)\\
\notag
R_1 + R_2 + R_3 &\geq \max_{j=1,2,3} I(X;A_1|Y_j) + \max_{j=2,3} I(X;A_2|A_1,Y_j) \\ 
&\hspace{45mm} + I(X;A_3|A_1,A_2,Y_3).
\end{align}
\end{subequations}
\end{enumerate}

\begin{lemma}\label{Lem:SR-Achievability-Gen}
The rates in $\set{R}_\text{in}(D_1,D_2,D_3)$ are all achievable; that is, 
\begin{equation} 
\set{R}_\text{in}(D_1,D_2,D_3) \subseteq \set{R}(D_1,D_2,D_3).
\end{equation}
\end{lemma} 

The next theorem, which is due to Tian and Diggavi~\cite{Tian-Aug-2007-A}, shows that the entire RD region is subsumed by $\set{R}_\text{in}(D_1,D_2,D_3)$ whenever the side information is physically degraded as in \eqref{Eqn:Thm:SR-Degraded-SI-Markov-Chain}.

\begin{theorem}\label{Thm:SR-Degraded-SI}
If the side information is physically degraded in the sense
\begin{equation}\label{Eqn:Thm:SR-Degraded-SI-Markov-Chain}
X \markov Y_3 \markov Y_2\markov Y_1,
\end{equation}
then~\cite[Thm.~1]{Tian-Aug-2007-A}
\begin{equation}
\set{R}_\text{in}(D_1,D_2,D_3) = \set{R}(D_1,D_2,D_3).
\end{equation}
Moreover, the rate constraints in~\eqref{Eqn:Lem:SR-Achievability-Gen-Rates} simplify to 
\begin{subequations}\label{Eqn:Thm:SR-Degraded-SI-Rates}
\begin{align}
R_1 & \geq I(X;A_1|Y_1)\\
R_1 + R_2 & \geq I(X;A_1|Y_1) + I(X;A_2|A_1,Y_2) \\
R_1 + R_2 + R_3 & \geq I(X;A_1|Y_1) + I(X;A_2|A_1,Y_2) + I(X;A_3|A_1,A_2,Y_3),
\end{align}
\end{subequations}
where $A_1$, $A_2$ and $A_3$ obey the cardinality constraints in~\eqref{Eqn:SR-Achievability-Gen-Card}, see also~\cite[Thm.~1]{Tian-Aug-2007-A}.
\end{theorem}

The achievability part of Theorem~\ref{Thm:SR-Degraded-SI} is given by Lemma~\ref{Lem:SR-Achievability-Gen}, and the simplified rate constraints in~\eqref{Eqn:Thm:SR-Degraded-SI-Rates} follow from the Markov chain~\eqref{Eqn:Thm:SR-Degraded-SI-Markov-Chain}. The converse assertion was proved by Tian and Diggavi in~\cite[App.~I]{Tian-Aug-2007-A} and there, again, the Markov chain~\eqref{Eqn:Thm:SR-Degraded-SI-Markov-Chain} enabled a crucial step. 

We now consider Theorem~\ref{Thm:SR-Degraded-SI} with conditionally less noisy side information and, as previously, deterministic distortion functions at Receivers 1 and~2. In particular, Receivers~1 and 2 wish to reconstruct almost losslessly
\begin{equation}
\tilde{X}_1 \triangleq \psi_1(X) \quad \text{and} \quad \tilde{X}_2 \triangleq \psi_2(X),
\end{equation}
respectively,
 where $\psi_1$ and $\psi_2$ are functions of the form
\begin{equation}
\psi_j : \set{X} \longrightarrow \tilde{\set{X}}_j, \quad j=1,2.
\end{equation}
Theorem~\ref{Thm:SR-Degraded-SI}, with deterministic $\delta_1$ and $\delta_2$, simplifies as follows. Define
\begin{equation*}
S'(D_3) \triangleq \min I(X;A_3|\tilde{X}_1,\tilde{X}_2,Y_3),\quad D_3 \geq 0,
\end{equation*}
where the minimisation is taken over all auxiliary $A_3$, jointly distributed with $(X,Y_1,Y_2,Y_3)$, such that the following is true:
\begin{enumerate}
\item[(i)] the auxiliary random variable is conditionally independent of the side information given $X$,
\begin{equation}
A_3\ \markov X \ \markov (Y_1,Y_2,Y_3);
\end{equation}
\item[(ii)] the cardinality of the alphabet of $A_3$ is bound by 
\begin{equation}
|\set{A}_3| \leq |\set{X}| + 1;
\end{equation}
\item[(iii)] there exists a (deterministic) map 
\begin{equation}
\phi_3 : \set{A}_3 \times \tilde{\set{X}}_1 \times \tilde{\set{X}}_2 \times \set{Y}_3\longrightarrow \hat{\set{X}}_3
\end{equation}
with 
\begin{equation}
D_3 \geq \mathbb{E}\ \delta_3\big(X,\phi_3(A_3,\tilde{X}_1,\tilde{X}_2,Y_3)\big).
\end{equation}
\end{enumerate}

\begin{corollary}
If the side information is physically degraded as in \eqref{Eqn:Thm:SR-Degraded-SI-Markov-Chain} and $\delta_1$ and $\delta_2$ are deterministic, then $\set{R}(0,0,D_3)$ is equal to the set of all rate tuples $(R_1,R_2,R_3)$ satisfying 
\begin{subequations}\label{Eqn:Cor9rates}
\begin{align}
R_1 & \geq H(\tilde{X}_1|Y_1) \\
R_1 + R_2 & \geq H(\tilde{X}_1|Y_1) + H(\tilde{X}_2 | \tilde{X}_1, Y_2) \\
R_1 + R_2 + R_3 & \geq H(\tilde{X}_1|Y_1) + H(\tilde{X}_2 | \tilde{X}_1, Y_2) + S'(D_3).
\end{align}
\end{subequations}
\end{corollary}

\begin{proof}
The achievability part follows directly from Theorem~\ref{Thm:SR-Degraded-SI} upon selecting the auxiliary random variables as $A_1 = \tilde{X}_1$ and $A_2 = \tilde{X}_2$ as well as recalling the definition of $S'(D_3)$. The converse can be proved following arguments similar to those used in Appendix~\ref{App:DeterministicDegradedSI} and is omitted for brevity. 
\end{proof}

The next lemma is a converse for arbitrarily distributed side information: it is a successive-refinement analogue of Lemma~\ref{Lem:Converse}. Let $\set{R}_\text{out}(D_3)$ denote the set of all rate tuples $(R_1,R_2,R_3)$ for which
\begin{align}
R_1 &\geq H(\tilde{X}_1|Y_1) \\
R_1 + R_2 &\geq H(\tilde{X}_1|Y_1) + H(\tilde{X}_2 | \tilde{X}_1, Y_2)
 + \min_W \Big\{I(W;Y_2|\tilde{X}_1) - I(W;Y_1 | \tilde{X}_1) \Big\}\\
\notag
R_1 + R_2 + R_3 & \geq H(\tilde{X}_1|Y_1) + H(\tilde{X}_2 | \tilde{X}_1, Y_2) + S'(D_3) + \min_W \Big\{I(W;Y_2|\tilde{X}_1) - I(W;Y_1 | \tilde{X}_1) \Big\}\\
&\hspace{5mm} + \min_W \Big\{I(W;Y_3|\tilde{X}_1,\tilde{X}_2) - I(W;Y_2 | \tilde{X}_1,\tilde{X}_2)\Big\},
\end{align}
where each minimisation is independently taken over all auxiliary $W$, jointly distributed with $(X,Y_1,Y_2,$ $Y_3)$, such that $|\set{W}| \leq |\set{X}|$ and $W \markov X \markov (Y_1,Y_2,Y_3)$. 

\begin{lemma}[Converse]\label{Lem:SR-Converse}
If $\delta_1$ and $\delta_2$ are deterministic, then
\begin{equation}
\set{R}_\text{out}(D_3) \supseteq \set{R}(0,0,D_3).
\end{equation}
\end{lemma}

Our proof of Lemma~\ref{Lem:SR-Converse} is quite similar to that of Lemma~\ref{Lem:Converse}, and it is given in Appendix~\ref{App:SR-Converse}. The next theorem shows that the outer bound (converse) of Lemma~\ref{Lem:SR-Converse} matches the inner bound (achievability) of Lemma~\ref{Lem:SR-Achievability-Gen} for a certain conditionally less noisy setting. 

\begin{theorem}
If $\delta_1$ and $\delta_2$ are deterministic, 
\begin{equation}
 \cln{Y_2}{Y_1}{\tilde{X}_1}\quad \text{and} \quad  \cln{Y_3}{Y_2}{\tilde{X}_1,\tilde{X}_2}, \label{Eqn:CLN3}
\end{equation} 
as well as
\begin{subequations}
\label{Eqn:condentropy23}
\begin{align}
 H(\tilde{X}_1 |  Y_1) &\geq \max\big\{H(\tilde{X}_1 |  Y_2) , H(\tilde{X}_1|Y_3)\big\}, \label{Eqn:condentropy2}\\
 H(\tilde{X}_2 | \tilde{X}_1, Y_2) & \geq H(\tilde{X}_2 | \tilde{X}_1 , Y_3),\label{Eqn:condentropy3}
\end{align}
\end{subequations}
then $\set{R}(0,0,D_3)$ is equal to the set of all rate tuples $(R_1,R_2,R_3)$ satisfying \eqref{Eqn:Cor9rates}, i.e.,
\begin{subequations}
\begin{align}
R_1 & \geq H(\tilde{X}_1|Y_1) \\
R_1 + R_2 & \geq H(\tilde{X}_1|Y_1) + H(\tilde{X}_2 | \tilde{X}_1, Y_2) \\
R_1 + R_2 + R_3 & \geq H(\tilde{X}_1|Y_1) + H(\tilde{X}_2 | \tilde{X}_1, Y_2) + S'(D_3).
\end{align}
\end{subequations}
\end{theorem}
\begin{proof}
The converse follows directly by Lemma~\ref{Lem:SR-Converse} and uses the conditionally less noisy assumptions~\eqref{Eqn:CLN3}. The achievability follows by Lemma~\ref{Lem:SR-Achievability-Gen} with  $A=\tilde{X}_1$ and $B=\tilde{X}_2$ and uses inequalities~\eqref{Eqn:condentropy23}. \end{proof}

\begin{remark}
Steinberg and Merhav~\cite{Steinberg-Aug-2004-A} were the first to consider and solve the two-stage successive refinement problem with physically degraded side information. Tian and Diggavi's work~\cite{Tian-Aug-2007-A} generalises Steinberg and Merhav's result to three or more stages with physically degraded side information. 
\end{remark}


\subsection{Two Stages with $Y_1$ better than $Y_2$ (abhinc $X \markov Y_1 \markov Y_2$)}\label{Sec:Scalable}

Reconsider the successive-refinement problem in Fig.~\ref{Fig:SR}, but now with only two receivers, Receiver~1 and 2. Moreover, suppose that the side information at Receiver~1 is better than the side information at Receiver~2. \emph{Side information scalable source coding} refers to the special case where 
\begin{equation}\label{Eqn:ScalableSI}
X\ \markov Y_1\ \markov Y_2.
\end{equation}
Notice that the roles of $Y_1$ and $Y_2$ in~\eqref{Eqn:ScalableSI} are reversed with respect to Definition~\ref{Def:DegradedSI} and Theorem~\ref{Thm:SR-Degraded-SI}. In contrast to Theorem~\ref{Thm:SR-Degraded-SI}, however, there is no known computable expression for the RD region in this setting. Tian and Diggavi give achievability and converse bounds in~\cite{Tian-Dec-2008-A}, and they show that these bounds match for degraded deterministic distortion measures. We wish to relax the Markov chain in~\eqref{Eqn:ScalableSI} to a conditionally less noisy setting and yet still recover the special case results of Tian and Diggavi.  

The next lemma gives an achievable rate region for arbitrarily distributed side information. Like in Lemma~\ref{Lem:SR-Achievability-Gen}, the rate constraints can be distilled from the rate constraints in~\cite{Timo-Aug-2011-A}, see Appendix~\ref{App:SR-SSC-Achievability-Gen}, and the cardinality bounds can be derived by the standard convex cover method. The lemma includes Tian and Diggavi's bound~\cite[Cor.~1]{Tian-Dec-2008-A} for arbitrarily distributed side information as a special case. 

Let $\set{R}^*_\text{in}(D_1,D_2)$ denote the set of all rate pairs $(R_1,R_2)$ for which there exist auxiliary random variables $(A_{12},A_1,A_2)$, jointly distributed with the source $(X,Y_1,Y_2)$, such that the following is true:
\begin{enumerate}
\item[(i)] there is a Markov chain,
\begin{equation}
(A_{12},A_1,A_2)\ \markov X\ \markov (Y_1,Y_2);
\end{equation}
\item[(ii)] the cardinalities of the alphabets of $A_{12}$, $A_1$ and $A_2$ respectively satisfy
\begin{align}
|\set{A}_{12}| &\leq |\set{X}| + 3\\
|\set{A}_1| &\leq |\set{X}|\ |\set{A}_{12}| + 1\\
|\set{A}_2| &\leq |\set{X}|\ |\set{A}_{12}| + 1;
\end{align}
\item[(iii)] there exist deterministic maps for $j = 1,2$,
\begin{equation}
\phi_j : \set{A}_j \times \set{Y}_j \longrightarrow \hat{\set{X}}_j,
\end{equation}
with
\begin{equation}
D_j \geq \mathbb{E}\ \delta_j \big(X,\phi_j(A_j,Y_j)\big);
\end{equation}
\item[(iv)] the rate pair $(R_1,R_2)$ satisfies 
\begin{subequations}\label{Eqn:SSC-Ach-1}
\begin{align}
R_1 &\geq I(X;A_{12},A_1|Y_1)\\
R_1 + R_2 &\geq \max \Big\{ I(X;A_{12}|Y_1), I(X;A_{12}|Y_2) \Big\} + I(X;A_1|A_{12},Y_1) + I(X;A_2|A_{12},Y_2).
\end{align}
\end{subequations}
\end{enumerate}

\begin{lemma}\label{Lem:SSC-Achievability-Gen}
The rate pairs in $\set{R}^*_\text{in}(D_1,D_2)$ are all achievable; that is,
\begin{equation}
\set{R}^*_\text{in}(D_1,D_2) \subseteq \set{R}(D_1,D_2).
\end{equation}
\end{lemma}

The next and final result of the paper generalises Tian and Diggavi's result~\cite[Thm.~4]{Tian-Dec-2008-A}, which holds under the Markov chain in~\eqref{Eqn:ScalableSI}, to a conditionally less noisy setting. Suppose $\delta_1$ and $\delta_2$ are deterministic, with $\tilde{X}_1 = \psi_1(X)$ and $\tilde{X}_2 = \psi_2(X)$. It is said that $\delta_2$ is a \emph{degraded version} of $\delta_1$ if 
\begin{equation}
\psi_2 = \psi' \circ \psi_{1}
\end{equation}
for some deterministic map $\psi'$. The next theorem is proved in Appendix~\ref{Proof:Thm:SSC-Det-Deg-1}.

\begin{theorem}\label{Thm:SSC-Det-Deg-1}
Suppose that $\delta_1$ and $\delta_2$ are deterministic. 
\begin{enumerate}
\item[(i)] If $\delta_2$ is a degraded version of $\delta_1$,
\begin{equation}
H(\tilde{X}_2|Y_1) \leq H(\tilde{X}_2|Y_2) \quad \text{and} \quad \cln{Y_1}{Y_2}{\tilde{X}_2},
\end{equation}
then $\set{R}^*_\text{in}(0,0) = \set{R}(0,0)$ and the rate constraints of~\eqref{Eqn:SSC-Ach-1} simplify to 
\begin{subequations}\label{Eqn:Proof-SSC-Det-Deg-1} 
\begin{align}
\label{Eqn:Proof-SSC-Det-Deg-1a}
R_1 &\geq H(\tilde{X}_1|Y_1) \\
\label{Eqn:Proof-SSC-Det-Deg-1b}
R_1 + R_2 & \geq H(\tilde{X}_2|Y_2) + H(\tilde{X}_1 | \tilde{X}_2,Y_1).
\end{align}
\end{subequations}

\item[(ii)] If $\delta_1$ is a degraded version of $\delta_2$ and
\begin{equation}
H(\tilde{X}_1|Y_1) \leq H(\tilde{X}_1|Y_2)
\end{equation}
then $\set{R}^*_\text{in}(0,0) = \set{R}(0,0)$ and the rate constraints of~\eqref{Eqn:SSC-Ach-1} simplify to  
\begin{subequations}\label{Eqn:Proof-SSC-Det-Deg-2} 
\begin{align}
R_1 &\geq H(\tilde{X}_1|Y_1)\\
R_1 + R_2 &\geq H(\tilde{X}_2|Y_2).
\end{align}
\end{subequations}
\end{enumerate}
\end{theorem}


\appendices

\section{Proof of Lemma~\ref{Lem:Conditional-Mathis-Lemma}}\label{App:Proof-Mathis-Lemmas}

\subsection{Preliminaries}
The proof will make use of the following telescoping identity. For any string of arbitrarily distributed random variables, $(A_1,B_1)$, $(A_2,B_2)$, $\ldots$, $(A_n,B_n)$, we have~\cite[Sec.~G]{Kramer-Dec-2011-A}
\begin{equation}\label{Eqn:Telescope}
\sum_{i=1}^n I(A_1^i;B_{i+1}^n) = \sum_{i=1}^n I(A_1^{i-1};B_i^n),
\end{equation}
with the notational conventions 
\begin{equation}
A_j^k \triangleq A_j,A_{j+1},\ldots,A_k\quad \text{and}\quad B_j^k \triangleq B_j,B_{j+1},\ldots,B_k
\end{equation}
for $1\leq j \leq k \leq n$ as well as
\begin{equation}
I(A_1^n;B_{n+1}^n)\triangleq 0\quad \text{ and } \quad I(A_1^{-1};B_0^n)\triangleq 0.
\end{equation}
These notations are used throughout the proof.

\subsection{Proof}

We first prove \eqref{Eqn:Mathis-Conditional-MI-Equality}. Notice that 
\begin{equation}\label{Eqn:Proof-Conditional-Mathis-Lemma-Step-1}
I(J;\bm{S}_2|\bm{L}) - I(J;\bm{S}_1|\bm{L}) 
= I(J;\bm{S}_2,\bm{L}) - I(J;\bm{S}_1,\bm{L}),
\end{equation}
by the chain rule for mutual information. Expand the first mutual information term $I(J;\bm{S_2},\bm{L})$ on the right hand side of~\eqref{Eqn:Proof-Conditional-Mathis-Lemma-Step-1} as follows:
\begin{align}
%
%
I(J;\bm{S_2},\bm{L}) &\step{a}{=}  \sum_{i=1}^n I(J; S_{2,i},L_i | S_{2,1}^{i-1}, L_1^{i-1}) \\
%
%
& \step{b}{=}  \sum_{i=1}^n  I(J, S_{2,1}^{i-1},L_1^{i-1}; S_{2,i},L_i)\\
%
%
\notag
& \step{c}{=}  \sum_{i=1}^n  \Big( I(J, S_{1,i+1}^n, S_{2,1}^{i-1},L_1^{i-1}, L_{i+1}^n   ; S_{2,i},L_i)\\
&\hspace{15mm} - I(S_{1,i+1}^n,L_{i+1}^n; S_{2,i},L_i | J, S_{2,1}^{i-1},L_1^{i-1})\Big) \\
& \step{d}{=}  \sum_{i=1}^n  \Big( I(W_i; S_{2,i},L_i) -  I(S_{1,i+1}^n,L_{i+1}^n; S_{2,i},L_i|J, S_{2,1}^{i-1},L_1^{i-1})\Big)\label{Eqn:ConverseTelescope1}
\end{align}
where (a) and (c) follow from the chain rule for mutual information; (b) exploits the fact that the source is i.i.d. and therefore  
\begin{equation}
H(S_{2,i},L_i | S_{2,1}^{i-1},L_1^{i-1}) = H(S_{2,i},L_i);
\end{equation}
and, finally, in (d) we define and substitute the random variable
\begin{equation}\label{Eqn:Define-Wi}
W_i \triangleq (J,S_{1,i+1}^n,S_{2,1}^{i-1},L_1^{i-1},L_{i+1}^n).
\end{equation}
 
Expand the second mutual information term $I(J;\bm{S_1},\bm{L})$ on the right hand side of~\eqref{Eqn:Proof-Conditional-Mathis-Lemma-Step-1} using the telescoping identity~\eqref{Eqn:Telescope} as follows:
\begin{align}
%
%
I(J;\bm{S_1},\bm{L}) &\step{a}{=} \sum_{i=1}^n \Big(I(J, S_{2,1}^{i-1},L_1^{i-1};S_{1,i}^n,L_i^n) - I(J,S_{2,1}^i,L_1^i;S_{1,i+1}^n,L_{i+1}^n)\Big)\\
%
%
\notag
&\step{b}{=}  \sum_{i=1}^n \Big( I(J,S_{2,1}^{i-1},L_1^{i-1};S_{1,i},L_i|S_{1,i+1}^n,L_{i+1}^n)\\
&\hspace{15mm} - I(S_{2,i},L_i;S_{1,i+1}^n,L_{i+1}^n | J,S_{2,1}^{i-1},L_1^{i-1})\Big) \\
%
%
\notag
&\step{c}{=} \sum_{i=1}^n \Big( I(J,S_{1,i+1}^n,S_{2,1}^{i-1},L_1^{i-1},L_{i+1}^n;S_{1,i},L_i) \\
&\hspace{15mm} - I(S_{2,i},L_i; S_{1,i+1}^n,L_{i+1}^n | J,S_{2,1}^{i-1},L_1^{i-1}) \Big)\\
\label{Eqn:ConverseTelescope2}
& \step{d}{=}\sum_{i=1}^n \Big( I(W_i;S_{1,i},L_i) 
- I(S_{2,i},L_i;S_{1,i+1}^n,L_{i+1}^n | J,S_{2,1}^{i-1},L_1^{i-1}) \Big),
\end{align} 
where (a) invokes the telescoping identity~\eqref{Eqn:Telescope} and the chain rule for mutual information; (b) again uses the chain rule for mutual information; (c) exploits the i.i.d. source and hence 
\begin{equation}
H(S_{1,i},L_i| S_{1,1}^{i-1},L_1^{i-1}) = H(S_{1,i},L_i);
\end{equation}
and, finally, in (d) we substitute $W_i \equiv (J,S_{1,i+1}^n,S_{2,1}^{i-1},L_1^{i-1},L_{i+1}^n)$.

Subtract~\eqref{Eqn:ConverseTelescope2} from~\eqref{Eqn:ConverseTelescope1} to obtain 
\begin{equation}\label{Eqn:Conversesum}
I(J;\bm{S_2},\bm{L}) - I(J;\bm{S_1},\bm{L})  =   \sum_{i=1}^n I(W_i;S_{2,i},L_i) -  I(W_i;S_{1,i},L_i). 
\end{equation}

We now \emph{single-letterize} the quantity on the right hand side of~\eqref{Eqn:Conversesum}. To this end, we introduce a time-sharing random variable: let $Q$ be uniform on $\{1,2,\ldots,n\}$ and independent of the tuple $(\bm{R},\bm{S_1},\bm{S_2},\bm{T},$ $\bm{L})$. Dividing~\eqref{Eqn:Conversesum} by $n$, we have 
\begin{align}
%
\frac{1}{n} \Bigg( \sum_{i=1}^n & I(W_i;S_{2,i},L_i) -   I(W_i;S_{1,i},L_i)\Bigg)\\
%
&\step{a}{=} 
\frac{1}{n} \sum_{i=1}^n \Big(I(W_i;S_{2,i},L_i|Q= i)-    I(W_i;S_{1,i},L_i|Q=i)\Big)\\
%
&\step{b}{=} I(W_Q;S_{2,Q},L_Q|Q)-    I(W_Q;S_{1,Q},L_Q|Q) \\
%
&\step{c}{=}I(W_Q,Q;S_{2,Q},L_Q) - I(W_Q, Q; S_{1,Q},L_Q) \\
\label{Eqn:ConverseTimeSharingStep1}
&\step{d}{=} I(\tilde{W};S_2,L)-I(\tilde{W};S_1,L),
\end{align} 
where in (a) we use that $Q$ is independent of $(S_{1,i},S_{2,i},L_i,W_i)$; in (b) that $Q$ is uniformly distributed; in (c) that $(\bm{S_1},\bm{S_2},\bm{L})$ is i.i.d. and independent of $Q$, and therefore 
\begin{equation}
H(S_{1,Q},L_Q|Q) = H(S_{1,Q},L_Q);
\end{equation}
and, finally, in (d) we define and substitute
\begin{equation}
\tilde{W} = (W_Q,Q),\ S_1=S_{1,Q},\ S_2 = S_{2,Q},\text{ and } L = L_Q.
\end{equation}

From~\eqref{Eqn:Conversesum} and \eqref{Eqn:ConverseTimeSharingStep1}, we have
\begin{equation}
I(J;\bm{S_2},\bm{L}) - I(J;\bm{S_1},\bm{L}) = n\big(I(\tilde{W};S_2,L) - I(\tilde{W};S_1,L) \big).
\end{equation}
We also notice that 
\begin{equation}\label{Eqn:Mathis-Markov-Chain}
W_i\ \markov (R_i,L_i)\ \markov (S_{1,i},S_{2,i},T_i),
\end{equation}
forms a Markov chain for all $i = 1,2,\ldots,n$. Each of the $n$ Markov chains in~\eqref{Eqn:Mathis-Markov-Chain} follows from the definition of $W_i$, the $n$-letter chain
\begin{equation}
J\ \markov (\bm{R},\bm{L})\ \markov (\bm{S_1},\bm{S_2},\bm{T}),
\end{equation}
and the fact that $(\bm{R}, \bm{S}_1, \bm{S_2},\bm{T},\bm{L})$ is i.i.d. Now define 
\begin{equation}
R=R_Q \quad \text{and}\quad T =T_Q.
\end{equation}
Using the independence of $Q$ from $(\bm{R}, \bm{T}, \bm{S_1}, \bm{S}_2,\bm{L})$, we have the desired Markov chain,
\begin{equation}\label{Eqn:Mathis-Markov-Chain-2}
\tilde{W}\ \markov (R,L)\ \markov (S_1,S_2,T).
\end{equation}

It remains to show that the auxiliary random variable $\tilde{W}$, whose alphabet cardinality is unbounded in $n$, can be replaced by some $W$ with an alphabet satisfying~\eqref{Eqn:Mathis-Conditional-Card}. We now prove the existence of such using the convex cover method of, for example, \cite[App.~C]{El-Gamal-2011-B}. 

For each and every $\tilde{w}$ in the support set of $\tilde{W}$, let $q_{\tilde{w}}$ denote the conditional distribution of $(R,S_1,S_2,$ $T,L)$ given $\tilde{W} = \tilde{w}$. Let $\set{P}$ denote the set of all joint distributions on $\set{R} \times \set{S}_1 \times \set{S}_2 \times \set{T} \times \set{L}$. 

For each and every pair $(r,l)$ in $\set{R} \times \set{L}$ but one --- the omitted pair, say $(r^*,l^*)$, can be chosen arbitrarily --- define the functional $g_{r,l} : \set{P} \longrightarrow [0,1]$,
\begin{equation}\label{Eqn:Mathis-Lemma-Grs1s2t}
g_{r,l}(q) \triangleq \sum_{s_1 \in \set{S}_1} \sum_{s_2 \in \set{S}_2} \sum_{t \in \set{T}} q(r,s_1,s_2,t,l).
\end{equation}
The $(|\set{R}| |\set{L}| - 1)$-functionals defined in~\eqref{Eqn:Mathis-Lemma-Grs1s2t} will be used to preserve the joint distribution of $(R,S_1,S_2,$ $T,L)$ when the Support Lemma~\cite[Sec. App.~C]{El-Gamal-2011-B} is invoked shortly. Indeed, we notice that for each such pair $(r,l)$ the expectation 
\begin{align}
\mathbb{E}_{\tilde{W}} \big\{g_{r,l}\big(q_{\tilde{W}}\big) \big\} \equiv \sum_{\tilde{w} \in \tilde{\set{W}}} \mathbb{P}[\tilde{W}=\tilde{w}]\ g_{r,l}(q_{\tilde{w}})
\end{align}
is equal to the true probability $\mathbb{P}[ (R,L) = (r,l)]$. Moreover, this agreement extends over $\set{R} \times \set{S}_1 \times \set{S}_2 \times \set{T} \times \set{L}$ because
\begin{equation}\label{Eqn:Proof:Cond-Mathis-Lemma-MC-1}
\mathbb{E}\big\{ g_{r,l}(q_{\tilde{W}}) \big\} \cdot \mathbb{P}\big[S_1 = s_1,S_2=s_2,T=t | R=r,L=l \big]
\end{equation}
is equal to the true joint probability $\mathbb{P}[R=r,S_1=s_1,S_2=s_2,T=t,L=l]$.

If the joint distribution of $(R,L,S_1,S_2,T)$ is preserved, we can additionally preserve the difference
\begin{equation}
I(\tilde{W};S_2,L) - I(\tilde{W};S_1,L)
\end{equation}
by simply preserving $H(S_2,L|\tilde{W}) - H(S_1,L|\tilde{W})$. To this end, define 
\begin{equation}
g(q) \triangleq H(\sf{S}_2,\sf{L}) - H(\sf{S}_1,\sf{L}),
\end{equation}
where the joint distribution\footnote{We use sans serif font to emphasise that this joint distribution differs to that of $(R,S_1,S_2,T,L)$.} of $(\sf{R},\sf{S}_1,\sf{S}_2,\sf{T},\sf{L})$ is understood to be given by $q$.
We also notice that
\begin{align}
\mathbb{E}_{\tilde{W}} \big\{ g(q_{\tilde{W}}) \big\} &\equiv \sum_{\tilde{w} \in \tilde{\set{W}}} \mathbb{P}[\tilde{W} = \tilde{w}] g(q_{\tilde{w}})\\
&= H(S_2,L|\tilde{W}) - H(S_1,L|\tilde{W}).
\end{align}

The Support Lemma asserts that there exists an auxiliary random variable $W$ defined on an alphabet $\set{W}$ with cardinality 
\begin{equation*}
|\set{W}| \leq |\set{R}| |\set{L}| 
\end{equation*}
and a collection of (conditional) joint distributions $\{q_w\}$ from $\set{P}$, indexed by the elements $w$ of $\set{W}$, such that 
\begin{enumerate}
\item[(i)] for all $(r,l)$ in $\set{R} \times \set{L}$ --- excluding the omitted pair $(r^*,l^*)$ --- we have
\begin{align}
\label{Eqn:Mathis-Lemma-Card-Cond-Dist}
\mathbb{E}_{W} \big\{ g_{r,l} (q_{W}) \big\}
&= \mathbb{E}_{\tilde{W}} \big\{ g_{r,l} (q_{\tilde{W}}) \big\}, 
\end{align}
\item[(ii)] and
\begin{align}
\label{Eqn:Mathis-Lemma-Card-MI}
\mathbb{E}_{W} \big\{ g (q_{W}) \big\}
&=\mathbb{E}_{\tilde{W}} \big\{ g (q_{\tilde{W}}) \big\}.
\end{align}
\end{enumerate}

The new auxiliary random variable $W$ and the distributions $\{q_w\}$ induce a joint distribution on $\set{W} \times \set{R} \times \set{L}$. The equality~\eqref{Eqn:Mathis-Lemma-Card-Cond-Dist} ensures that the $(R,L)$-marginal of this new distribution is equal to the true distribution of $(R,L)$. This agreement extends to the full joint distribution via~\eqref{Eqn:Proof:Cond-Mathis-Lemma-MC-1}; i.e., we impose the Markov chain
\begin{equation}
W\ \markov (X,L)\ \markov (S_1,S_2,T).
\end{equation} 
Finally, the equalities~\eqref{Eqn:Mathis-Lemma-Card-Cond-Dist}  and~\eqref{Eqn:Mathis-Lemma-Card-MI} imply 
\begin{equation}
I(W;S_2,L) - I(W;S_1,L) = I(\tilde{W};S_2,L) - I(\tilde{W}; S_1,L).
\end{equation}
\hfill $\blacksquare$

\medskip

\begin{remark}
\begin{enumerate}
\item[]
\item[(i)] A consequence of the telescoping identity~\eqref{Eqn:Telescope} is the classic \emph{Csisz\'{a}r sum} identity~\cite[Sec.~2.4]{El-Gamal-2011-B},
\begin{equation}\label{Eqn:CsiszarSum}
\sum_{i=1}^n I(A_i;B_{i+1}^n|A_1^{i-1}) = \sum_{i=1}^n I(B_i;A_1^{i-1}|B_{i+1}^n).
\end{equation}
The proof of Lemma~\ref{Lem:Conditional-Mathis-Lemma} can be manipulated so as to replace the \emph{telescoping sum identity} step~\eqref{Eqn:ConverseTelescope2} with a \emph{Csisz\'{a}r sum identity} step. We feel that the telescoping approach gives a cleaner proof. 

\item[(ii)] We note that steps (a) and (b) of~\eqref{Eqn:ConverseTelescope2} are reminiscent of those used in Kramer's converse for the \emph{Gelfand-Pinsker problem} (coding for channels with state), see~\cite[Sec.~F]{Kramer-Dec-2011-A} or \cite[Sec.~6.6]{Kramer-2008-A}. It is not clear, as yet, whether there is a deeper relationship between the two problems.
\end{enumerate}
\end{remark}


\section{Proof of Cardinality Bound~\eqref{Eqn:HB-Achievability-Card-Bounds} of Lemma~\ref{Lem:HB}}\label{App:HBcardinality}
\renewcommand{\d}{\textnormal{d}}

Suppose that we have auxiliary random variables $(A,B,C)$ as well as functions $\phi_1$ and $\phi_2$ that satisfy the Markov chain~\eqref{Eqn:HB-Achievability-Markov-Chain} and the average distortion condition~\eqref{Eqn:distortions}, but not the cardinality bounds~\eqref{Eqn:HB-Achievability-Card-Bounds}; i.e., the alphabets $\set{A}$, $\set{B}$ and $\set{C}$ are finite but otherwise arbitrary. 

Consider the variable $C$. For each and every $c$ in the support set of $C$, let $q_c$ denote the conditional distribution of $(A,B,X)$ given $C = c$. Let $\set{P}_1$ denote the set of all joint distributions on $\set{A} \times \set{B} \times \set{X}$. 

For each and every $x$ in $\set{X}$ but one, say $x^*$, define $g_x : \set{P}_1 \longrightarrow [0,1]$ by setting 
\begin{equation}
g_{x}(q)\triangleq \sum_{a \in \set{A}} \sum_{b \in \set{B}} q(a,b,x).
\end{equation}
We notice that, for all $x$ except $x^*$,
\begin{equation}\label{gx2}
\mathbb{E}_C\big\{g_x(q_C)\big\} = \mathbb{P}[X=x]
\end{equation}
gives the true marginal distribution of $X$. Now define the following functionals --- each mapping $\set{P}_1$ to $[0,\infty]$ --- by setting
\begin{align}
\label{Eqn:HB-card-bound-C-part-1}
g_1(q)&\triangleq I(\sf{X};\sf{B}|\sf{Y}_2)-H(\sf{X}|\sf{A},\sf{Y}_1)\\
\label{Eqn:HB-card-bound-C-part-2}
g_2(q)&\triangleq I(\sf{X};\sf{A}|\sf{Y}_1)-H(\sf{X}|\sf{B},\sf{Y}_2)\\
g_3(q)&\triangleq \sum_{a \in \set{A}} \sum_{y_1 \in \set{Y}_1}
\min\limits_{\hat{x}\in\hat{\set{X}}_1}
\sum_{b \in \set{B}} \sum_{x \in \set{X}} \sum_{y_2 \in \set{Y}_2} 
q(a,b,x)p(y_1,y_2|x)\delta_1(\hat{x},x)\\
g_4(q)&\triangleq \sum_{b \in \set{B}} \sum_{y_2 \in \set{Y}_2}
\min_{\hat{x}\in\hat{\set{X}}_2} 
\sum_{a \in \set{A}} \sum_{x \in \set{X}} \sum_{y_1 \in \set{Y}_1}
q(a,b,x)p(y_1,y_2|x)\delta_2(\hat{x},x),
\end{align}
where the joint distribution of $(\sf{A},\sf{B},\sf{X},\sf{Y}_1,\sf{Y}_2)$ in~\eqref{Eqn:HB-card-bound-C-part-1} and~\eqref{Eqn:HB-card-bound-C-part-2} is understood as follows: $(\sf{A},\sf{B},\sf{X})$ is distributed according to $q$ and $(\sf{Y}_1,\sf{Y}_2)$ conditionally depends on $X$ via the true side information channel (i.e., the conditional distribution $\mathbb{P}[Y_1=y_1,Y_2=y_2|X=x]$); in particular, we have imposed the Markov chain $(\sf{A},\sf{B}) \markov \sf{X} \markov (\sf{Y}_1,\sf{Y}_2)$. We also notice that 
\begin{align}
\mathbb{E}_C\big\{g_1(q_{C})\big\}
&= I(X;B|Y_2,C)-H(X|A,C,Y_1) \label{g1}\\
\mathbb{E}_C\big\{g_2(q_{C})\big\}
&=I(X;A|Y_1,C)-H(X|B,C,Y_2)\label{g2}\\
\mathbb{E}_C\big\{g_3(q_{C})\big\}
&=\!\!\min_{\phi_1:\set{A}\times\set{C}\times\set{Y}_1\to\hat{\set{X}_1}}\!\!\mathbb{E}\,\delta_1\big(X,\phi_1(A,C,Y_1)\big)\label{g3}\\
\mathbb{E}_C\big\{g_4(q_{C})\big\}
&=\!\!\min_{\phi_2:\set{B}\times\set{C}\times\set{Y}_2\to\hat{\set{X}_2}}\!\!\mathbb{E}\,\delta_2\big(X,\phi_2(B,C,Y_2)\big)\label{g4}.
\end{align}

The Support Lemma asserts that there exists a new auxiliary random variable $C^\dag$ defined on an alphabet $\set{C}^\dag$ with cardinality 
\begin{equation}
|\set{C}^\dag|\leq |\set{X}|+3
\end{equation}
together with a collection of $|\set{C}^\dag|$ distributions $\{q^\dag_{c}\}$ from $\set{P}_1$ --- indexed by the elements $c$ of $\set{C}^\dag$ --- such that 
\begin{align}
\label{gx}
\mathbb{E}_C\big\{g_x(q_C)\big\}&=\mathbb{E}_{C^\dag}\big\{g_x(q^\dag_{C^\dag})\big\}, \quad \forall x \in \set{X}\ \text{except}\ x^*
\end{align}
and
\begin{equation}\label{rg}
\mathbb{E}_C\big\{g_j(q_{C})\big\} = \mathbb{E}_{C^\dag}\big\{g_j(q^\dag_{C^\dag})\big\},\quad \forall j = 1,2,3,4.
\end{equation}

The new variable $C^\dag$, the distributions $\{q_c^\dag\}$, and the true side information channel come together via the Markov chain 
\begin{equation}\label{Eqn:HB-Card-Proof-Step1-Markov}
(A^\dag,B^\dag,C^\dag) \markov X^\dag \markov (Y_1^\dag,Y_2^\dag)
\end{equation}
to specify a tuple $(A^\dag,B^\dag,C^\dag,X^\dag,Y_1^\dag,Y_2^\dag)$ on $\set{A} \times \set{B} \times \set{C}^\dag \times \set{X} \times \set{Y}_1 \times \set{Y}_2$. The equality~\eqref{gx} ensures that $(X^\dag,Y^\dag_1,Y^\dag_2)$ and $(X,Y_1,Y_2)$ have the same distribution, which also implies
\begin{equation}\label{rrg1}
H(X^\dag|Y^\dag_1) = H(X|Y_1) \quad \text{and}\quad H(X^\dag|Y^\dag_2) = H(X|Y_2).
\end{equation}
Similarly,~\eqref{rg} ensures  
\begin{subequations}\label{rrg}
\begin{align}
I(X^\dag;B^\dag|Y^\dag_2,C^\dag)-H(X^\dag|B^\dag,C^\dag,Y^\dag_1) &= I(X;B|Y_2,C)-H(X|B,C,Y_1) \label{g1}\\
I(X^\dag;A^\dag|Y^\dag_1,C^\dag)-H(X^\dag|A^\dag,C^\dag,Y^\dag_2)
&=I(X;A|Y_1,C)-H(X|A,C,Y_2)\label{g2};
\end{align}
\end{subequations}
and
\begin{subequations}\label{ddg}
\begin{align}
\min_{\phi_1^\dag:\set{A}\times\set{C}^\dag\times\set{Y}_1\to\hat{\set{X}_1}}\!\!\mathbb{E}\,\delta_1\big(X^\dag,\phi_1^\dag(A^\dag,C^\dag,Y^\dag_1)\big)
&= 
\min_{\phi_1:\set{A}\times\set{C}\times\set{Y}_1\to\hat{\set{X}_1}}
\mathbb{E}\,\delta_1\big(X,\phi_1(A,C,Y_1)\big)\label{g3}\\
\min_{\phi_2^\dag:\set{B}\times\set{C}^\dag\times\set{Y}_2\to\hat{\set{X}_2}} 
\mathbb{E}\,\delta_2\big(X^\dag,\phi_2^\dag(B^\dag,C^\dag,Y^\dag_2)\big)
&=
\min_{\phi_2:\set{B}\times\set{C}\times\set{Y}_2\to\hat{\set{X}_2}}
\mathbb{E}\,\delta_2\big(X,\phi_2(B,C,Y_2)\big)\label{g4}.
\end{align}
\end{subequations}
Finally, the equalities \eqref{rrg1} and~\eqref{rrg} together give 
\begin{multline}\label{Eqn:card-bound-sum-rate-1}
\max_{j=1,2} I(X^\dag;C^\dag|Y^\dag_j) + I(X^\dag;A^\dag|C^\dag,Y^\dag_1) + I(X^\dag;B^\dag|C^\dag,Y^\dag_2)\\ 
=
\max_{j=1,2} I(X;C|Y_j) + I(X;A|C,Y_1) + I(X;B|C,Y_2).
\end{multline}

Consider the tuple $(A^\dag,B^\dag,C^\dag,X^\dag,Y_1^\dag,Y_2^\dag)$. We have the Markov chain~\eqref{Eqn:HB-Card-Proof-Step1-Markov} by construction, and we notice that $A^\dag$ and $B^\dag$ always appear separately in~\eqref{rrg} and \eqref{ddg}. We may therefore replace the joint distribution of $(A^\dag,B^\dag,C^\dag,X^\dag,Y^\dag_1,Y^\dag_2)$ with another that shares the same Markov chain~\eqref{Eqn:HB-Card-Proof-Step1-Markov} and marginals $(A^\dag,C^\dag,X^\dag)$, $(B^\dag,C^\dag,X^\dag)$ and $(X^\dag,Y^\dag_1,Y^\dag_2)$, but imposes the new chain
\begin{equation}\label{g-Markov}
A^\dag\ \markov (C^\dag,X^\dag)\ \markov B^\dag.
\end{equation}
Or put another way, the Markov chain~\eqref{g-Markov} does not alter the left hand sides of~\eqref{rrg} or \eqref{ddg}. The chain~\eqref{g-Markov} will be important in the sequel because it allows the cardinalities of $\set{A}$ and $\set{B}$ to be bound independently. With a slight abuse of notation, we retain the same notation $(A^\dag,B^\dag,C^\dag,X^\dag,Y^\dag_1,Y^\dag_2)$ for this new distribution.  

Consider the variable $A^\dag$. For each and every $a$ in the support set of $A^\dag$, let $q_a$ denote the conditional distribution of $(C^\dag,X^\dag)$ given $A^\dag = a$. Let $\set{P}_2$ denote the set of all joint distributions on $\set{C}^\dag \times \set{X}$. For each and every $(c,x)$ in $\set{C}^\dag \times \set{X}$ but one, define $g_{c,x} : \set{P}_2 \longrightarrow [0,1]$ by setting
\begin{equation}
g_{c,x}(q) \triangleq  q(c,x).
\end{equation}
Here $\mathbb{E}_{A^\dag} \big\{g_{c,x}(q_{A^\dag})\big\} = \mathbb{P}[(C^\dag,X^\dag) = (c,x)]$ returns the desired probability for all $(c,x)$ in $\set{C}^\dag \times \set{X}$ but one. In addition, define
\begin{equation}
g_5(q)\triangleq H(\sf{X}|\sf{C},\sf{Y}_1)
\end{equation}
and
\begin{equation}
g_6(q)\triangleq \sum_{c \in \set{C}^\dag} \sum_{y_1 \in \set{Y}_1} 
\min\limits_{\hat{x}\in\hat{\set{X}}_1}
\sum_{x \in \set{X}}\sum_{y_2 \in \set{Y}_2}q(c,x)p(y_1,y_2|x)\delta_1(\hat{x},x),
\end{equation}
where the joint distribution of $(\sf{C},\sf{X},\sf{Y}_1,\sf{Y}_2)$ is understood as follows: $(\sf{C},\sf{X})$ is distributed according to $q$, and $(\sf{Y}_1,\sf{Y}_2)$ conditionally depends on $\sf{X}$ via the true side information channel. We have
\begin{equation}
\mathbb{E}_{A^\dag}\big\{g_5(q_{A^\dag})\big\} = H(X^\dag | A^\dag,C^\dag,Y_1^\dag).
\end{equation}
and 
\begin{equation}
\mathbb{E}_{A^\dag}\big\{g_6(q_{A^\dag})\big\} = \min_{\phi_1^\dag : \set{A} \times \set{C}^\dag \times \set{Y}_1 \to \hat{\set{X}}_1} \mathbb{E}\delta_1\big(X,\phi_1^\dag(A^\dag,C^\dag,Y_1^\dag)\big).
\end{equation}

The Support Lemma asserts that there exists a random variable $A^\ddag$ defined on an alphabet $\set{A}^\ddag$ with cardinality 
\begin{equation}
|\set{A}^\ddag| \leq |\set{C}^\dag||\set{X}| + 1
\end{equation}
together with a collection of $|\set{A}^\ddag|$ distributions $\{q_a^\ddag\}$ from $\set{P}_2$ --- indexed by the elements $a$ of $\set{A}^\ddag$ --- such that 
\begin{equation}\label{Eqn:HB-Card-Bounds-A}
\mathbb{E}_{A^\ddag} \big\{ g_{c,x}(q_{A^\ddag}) \big\} = \mathbb{E}_{A^\dag} \big\{ g_{c,x}(q_{A^\dag}) \big\} 
\end{equation}
and 
\begin{equation}\label{Eqn:HB-Card-Bounds-A1}
\mathbb{E}_{A^\ddag} \big\{g_j(q_{A^\ddag}) \big\} = \mathbb{E}_{A^\dag} \big\{g_j(q_{A^\dag})\big\},\quad j = 5,6. 
\end{equation}
The new variable $A^\ddag$, the distributions $\{q_a^\ddag\}$, the true side information channel, the conditional distribution $P(B^\dag|X^\dag, C^\dag)$,  and the Markov chains~\eqref{Eqn:HB-Card-Proof-Step1-Markov} and~\eqref{g-Markov} come together to specify a tuple $(A^\ddag,B^\ddag,C^\ddag,X^\ddag,Y_1^\ddag,Y_2^\ddag)$ on $\set{A}^\ddag \times \set{B} \times \set{C}^\dag \times \set{X} \times \set{Y}_1 \times \set{Y}_2$.

The equalities in~\eqref{Eqn:HB-Card-Bounds-A} ensure that $(C^\ddag,X^\ddag)$ and $(C^\dag,X^\dag)$ have the same distribution. By construction, we also have that $(B^\ddag,C^\ddag,X^\ddag,Y_1^\ddag,Y_2^\ddag)$ and $(B^\dag,C^\dag,X^\dag,Y_1^\dag,Y_2^\dag)$ have the same distribution, and therefore
\begin{multline}\label{Eqn:HB-card-bound-step-2-rates}
\max\Big\{ I(X^\ddag;C^\ddag|Y^\ddag_1), I(X^\ddag;C^\ddag|Y^\ddag_2) \Big\} + H(X^\ddag|C^\ddag,Y^\ddag_1) + I(X^\ddag;B^\ddag|C^\ddag,Y^\ddag_2) \\
= \max\Big\{ I(X^\dag;C^\dag|Y^\dag_1), I(X^\dag;C^\dag|Y^\dag_2) \Big\} + H(X^\dag|C^\dag,Y^\dag_1) + I(X^\dag;B^\dag|C^\dag,Y^\dag_2). 
\end{multline} 
In addition, \eqref{Eqn:HB-Card-Bounds-A1} ensures that 
\begin{equation}\label{Eqn:HB-card-bound-step-2-rates2}
H(X^\ddag|A^\ddag,C^\ddag,Y^\ddag_1) = H(X^\dag|A^\dag,C^\dag,Y^\dag_1)
\end{equation}
and 
\begin{equation}\label{Eqn:HB-card-bound-step-2-dist1}
\min_{\phi_1^\ddag:\set{A}^\ddag\times\set{C}^\dag\times\set{Y}_1\to\hat{\set{X}_1}}\!\!\mathbb{E}\,\delta_1\big(X^\ddag,\phi_1^\ddag(A^\ddag,C^\ddag,Y^\ddag_1)\big)
= 
\min_{\phi^\dag_1:\set{A}\times\set{C}^\dag\times\set{Y}_1\to\hat{\set{X}_1}}
\mathbb{E}\,\delta_1\big(X^\dag,\phi_1(A^\dag,C^\dag,Y^\dag_1)\big).
\end{equation}
Combining~\eqref{Eqn:card-bound-sum-rate-1}, \eqref{ddg}, \eqref{Eqn:HB-card-bound-step-2-rates}, \eqref{Eqn:HB-card-bound-step-2-rates2} and~\eqref{Eqn:HB-card-bound-step-2-dist1} gives 
\begin{multline}\label{Eqn:card-bound-cond1}
\max\Big\{ I(X^\ddag;C^\ddag|Y^\ddag_1), I(X^\ddag;C^\ddag|Y^\ddag_2) \Big\} + I(X^\ddag;A^\ddag|C^\ddag,Y^\ddag_1) + I(X^\ddag;B^\ddag|C^\ddag,Y^\ddag_2) \\
= \max\Big\{ I(X;C|Y_1), I(X;C|Y_2) \Big\} + I(X;A|C,Y_1) + I(X;B|C,Y_2). 
\end{multline}
and
\begin{subequations}\label{Eqn:card-bound-cond2}
\begin{align}
\min_{\phi_1^\ddag:\set{A}^\ddag\times\set{C}^\dag\times\set{Y}_1\to\hat{\set{X}_1}}\!\!\mathbb{E}\,\delta_1\big(X^\ddag,\phi_1^\ddag(A^\ddag,C^\ddag,Y^\ddag_1)\big)
&= 
\min_{\phi_1:\set{A}\times\set{C}\times\set{Y}_1\to\hat{\set{X}_1}}
\mathbb{E}\,\delta_1\big(X,\phi_1(A,C,Y_1)\big)\label{g3}\\
\min_{\phi_2^\ddag:\set{B}\times\set{C}^\dag\times\set{Y}_2\to\hat{\set{X}_2}} 
\mathbb{E}\,\delta_2\big(X^\ddag,\phi_2^\ddag(B^\ddag,C^\ddag,Y^\ddag_2)\big)
&=
\min_{\phi_2:\set{B}\times\set{C}\times\set{Y}_2\to\hat{\set{X}_2}}
\mathbb{E}\,\delta_2\big(X,\phi_2(B,C,Y_2)\big)\label{g4},
\end{align}
\end{subequations}
as desired.

Using analogous arguments as above, we can find a random vector $(A', B', C', X', Y_1', Y_2')$ over $\set{A}^\ddag \times \set{B}' \times \set{C}^\dag \times \set{X} \times \set{Y}_1 \times \set{Y}_2$, where the cardinality of the alphabet $\set{B}'$ satisfies \begin{equation}
|\set{B}'| \leq |\set{C}^{\dag}| |\set{X}| +1,
\end{equation}
and such that \eqref{Eqn:card-bound-cond1} and \eqref{Eqn:card-bound-cond2} are satisfied when the tuple $(A^\ddag, B^\ddag, C^\ddag, X^\ddag, Y_1^\ddag, Y_2^\ddag)$ is replaced by the new tuple $(A', B', C', X', Y_1', Y_2')$. This concludes the proof of the cardinality bounds.\hfill $\blacksquare$


\section{Proof of Lemma~\ref{Lem:LessNoisySpecialCases}}\label{App:LessNoisySpecialCases}

\subsection{Assertion (i)}
Consider any auxiliary random variable $W$ for which 
\begin{equation}\label{Eqn:LessNoisySpecialCases-MC}
W \markov (X,L) \markov (Y_1,Y_2)
\end{equation}
is a Markov chain. We have
\begin{align}
I(W;Y_2|L) &= H(W|L) - H(W|L,Y_2)\\
&\step{a}{=} H(W|L) - H(W|L,Y_2,Y_1)\\
&\geq H(W|L) - H(W|L,Y_1)\\
&= I(W;Y_1|L),
\end{align}
where (a) uses the fact that 
\begin{equation}
W\ \markov (Y_2,L)\ \markov Y_1,
\end{equation}
which follows from~\eqref{Eqn:LessNoisySpecialCases-MC}, the Markov chain~\eqref{Eqn:LessNoisySpecialCases-MCinLemma1}, and the fact that the side information is physically degraded. \hfill $\blacksquare$

\subsection{Assertion (ii)}
Take any auxiliary random variable $W$ for which 
\begin{equation}
W\ \markov (X_1,X_2)\ \markov (Y_1,Y_2).
\end{equation}
Consider Definition~\ref{Def:ConditionallyLessNoisy} with $L = X_1$. We have
\begin{align}
%
0 & \leq I(W;Y_1|X_1)\\
%
& = H(Y_1|X_1) - H(Y_1|W,X_1)\\
%
& \step{a}{=} H(Y_1|X_1,X_2) - H(Y_1|W,X_1)\\
%
& \step{b}{=} H(Y_1|X_1,X_2) - H(Y_1|W,X_1,X_2)\\
%
& = I(W;Y_1|X_1,X_2)\\
& \step{c}{=} 0,
\end{align}
where the indicated steps apply the following Markov chains:
\begin{equation}
\begin{array}{rc}
\text{(a\hspace{0.2mm})} &  X_2 \markov X_1 \markov Y_1 \\
\text{(b)} &  X_2 \markov (W,X_1) \markov Y_1\\
\text{(c\hspace{0.2mm})} & W \markov (X_1,X_2) \markov (Y_1,Y_2).
\end{array}
\end{equation}
Thus, we have that 
\begin{equation}
I(W;Y_1|X_1)=0
\end{equation}
and therefore $I(W;Y_1|X_1)$ is no larger than $I(W;Y_2|X_1)$. \hfill $\blacksquare$



\section{Proof of Lemma~\ref{Lem:Converse}}\label{Sec:ConverseProof}

Let 
\begin{equation}\label{Eqn:tildeYProbError}
P_{\text{e},i} \triangleq \mathbb{P}\big[\hat{X}_{1,i} \neq \tilde{X}_i\big]
\end{equation}
denote the probability that the $i$-th symbol $\tilde{X}_i \equiv \psi(X_i)$ is reconstructed in error at Receiver~1. The probability $P_{\text{e},i}$ can also be expressed as $P_{\text{e},i} = \mathbb{E} \delta_1(X_i,\hat{X}_{1,i})$ and, therefore, we have
\begin{equation}\label{Eqn:tildeYAvgProbError}
\frac{1}{n} \sum_{i=1}^n P_{\text{e},i} \leq \epsilon
\end{equation}
 from the definition of achievability. Consider the conditional entropy $H(\bm{\tilde{X}}|M,\bm{Y_1})$. Starting from the fact that $\bm{\hat{X}_1}$ is determined by $(M,\bm{Y_1})$, we have
\begin{align}
%
%
H(\bm{\tilde{X}}|M,\bm{Y_1}) 
&\step{a}{=} H(\bm{\tilde{X}}|M,\bm{Y_1},\bm{\hat{X}_1})\\
%
%
&\leq H(\bm{\tilde{X}}|\bm{\hat{X}_1})\\
%
%
&\step{b}{\leq} \sum_{i=1}^n H(\tilde{X}_i|\hat{X}_{1,i})\\
%
&\step{c}{\leq} \sum_{i=1}^n \big( h(P_{\text{e},i}) + P_{\text{e},i} \log | \tilde{\set{X}} | \big)\\
%
&\step{d}{\leq} h\left(\sum_{i=1}^n P_{\text{e},i}\right) + \left( \sum_{i=1}^n P_{\text{e},i} \right) \log | \tilde{\set{X}} | \\
%
&\step{e}{\leq} n h(\epsilon) + n \epsilon \log |\tilde{\set{X}} |\\
\label{Eqn:tildeYFano}
&\step{f}{=} n\varepsilon(n,\epsilon),
\end{align}
where (a) applies the Markov chain
\begin{equation}
\bm{\tilde{X}} \markov (M,\bm{Y_1}) \markov \bm{\hat{X}_1};
\end{equation}
(b) invokes the chain rule for entropy and the fact that conditioning cannot increase entropy; (c) applies Fano's inequality; (d) combines the concavity of the binary entropy function with Jensen's inequality; (e)  invokes~\eqref{Eqn:tildeYAvgProbError}; and (f) substitutes 
\begin{equation}
\varepsilon(n,\epsilon) \triangleq h(\epsilon) + \epsilon \log |\tilde{\set{X}}|.
\end{equation}
Finally, we notice that $\varepsilon(n,\epsilon) \rightarrow 0$ as $\epsilon \rightarrow 0$.

Now consider the rate condition~\eqref{Eqn:Achievable-Rate}. We have
\begin{align}
%
R + \epsilon &\geq \frac{1}{n} \log_2 |\set{M}|\\
%
&\geq \frac{1}{n} H(M)\\
%
&\geq \frac{1}{n} H(M|\bm{Y_1})\\
%
&\geq \frac{1}{n} I(\bm{X},\bm{\tilde{X}};M|\bm{Y_1})\\
%
&= \frac{1}{n} \Big(I(\bm{\tilde{X}};M|\bm{Y_1}) + I(\bm{X};M|\bm{\tilde{X}},\bm{Y_1}) \Big)\\
%
&\step{a}{\geq} \frac{1}{n}\Big(H(\bm{\tilde{X}}|\bm{Y_1}) - n\varepsilon(n,\epsilon) 
+ I(\bm{X};M|\bm{\tilde{X}},\bm{Y_1}) \Big)\\
\label{Eqn:ConverseStep1}
&\step{b}{=} H(\tilde{X}|Y_1) - \varepsilon(n,\epsilon) + \frac{1}{n}I(\bm{X};M| \bm{\tilde{X}},\bm{Y_1}),
\end{align}
where (a) substitutes~\eqref{Eqn:tildeYFano} and (b) invokes the fact that $(\bm{X},\bm{\tilde{X}},\bm{Y_1})$ is i.i.d.

Consider the conditional mutual information term on the right hand side of~\eqref{Eqn:ConverseStep1}. Rearranging this term, with the intent of conditioning on $(\bm{\tilde{X}},\bm{Y_2})$ instead of $(\bm{\tilde{X}},\bm{Y_1})$, we obtain
\begin{align}
\notag
I(\bm{X} ; M | \bm{\tilde{X}},\bm{Y_1})
%
%
&\step{a}{=} I(\bm{X};M|\bm{\tilde{X}}, \bm{Y_2})- H(M|\bm{\tilde{X}}, \bm{Y_2}) +H(M|\bm{\tilde{X}}, \bm{Y_1}) \\
\label{Eqn:ConverseStep2}
&= I(\bm{X};M|\bm{\tilde{X}},\bm{Y_2}) + I(M;\bm{Y_2}|\bm{\tilde{X}}) - I(M;\bm{Y_1}|\bm{\tilde{X}})
\end{align}
where (a) invokes that $M$ is a function of $\bm{X}$ or, in the more general case of stochastic encoders, that 
\begin{equation}
M \markov \bm{X} \markov (\bm{\tilde{X}},\bm{Y_1},\bm{Y_2}).
\end{equation}

Consider the first conditional mutual information on the right hand side of~\eqref{Eqn:ConverseStep2}. 
Expand this term using the method of Wyner and Ziv~\cite[Eqn.~(52)]{Wyner-Jan-1976-A} as follows:

\begin{align}
%
%
I(\bm{X};M|\bm{\tilde{X}},\bm{Y_2}) &= \sum_{i=1}^n I(X_i ; M | \bm{\tilde{X}}, \bm{Y_2}, X_1^{i-1}) \\
%
%
&\step{a}{=} \sum_{i=1}^n I(X_i ; M, \tilde{X}_1^{i-1}, \tilde{X}_{i+1}^n, Y_{2,1}^{i-1}, Y_{2,i+1}^n, X_1^{i-1} | \tilde{X}_i, Y_{2,i})\\
%
%
&\geq \sum_{i=1}^n I(X_i ; M, Y_{2,1}^{i-1}, Y_{2,i+1}^n | \tilde{X}_i, Y_{2,i})\\
\label{Eqn:ConverseStep3}
&\step{b}{=} \sum_{i=1}^n I(X_i ; B_i | \tilde{X}_i, Y_{2,i}), 
\end{align}
where (a) follows because $(\bm{X},\bm{Y_2},\bm{\tilde{X}})$ i.i.d. and therefore 
\begin{equation}
H(X_i | \bm{\tilde{X}}, \bm{Y_2}, X_1^{i-1}) = H(X_i | \tilde{X}_i, Y_{2,i}), 
\end{equation}
and in (b) we define
\begin{equation}
B_i \triangleq (M,Y_{2,1}^{i-1},Y_{2,i+1}^n). 
\end{equation}

Continuing on from~\eqref{Eqn:ConverseStep3}, we have
\begin{align}
%
%
\frac{1}{n} I(\bm{X} ; M | \bm{\tilde{X}}, \bm{Y_2}) 
&\geq \frac{1}{n} \sum_{i=1}^n I(X_i ; B_i | \tilde{X}_i, {Y_{2,i}}) \\
%
%
&\step{a}{\geq} 
\frac{1}{n} \sum_{i=1}^n S\big( \mathbb{E} \delta_2(X_i , \hat{X}_{2,i}) \big)\\
%
%
&\step{b}{\geq} 
S\left(\mathbb{E} \frac{1}{n} \sum_{i=1}^n \delta_2(X_i,\hat{X}_{2,i}) \right)\\
\label{Eqn:ConverseStep4}
&\step{c}{\geq} S(D_2 + \epsilon), 
\end{align}
where 
\begin{itemize}
\item[(a)] follows from the definition of $S(D_2)$ upon noticing that the $i$-th reconstructed symbol, {$\hat{X}_{2,i}$}, can be expressed as a deterministic function of $(B_i,Y_{2,i})$ and
\begin{equation}
B_i \markov X_i \markov (Y_{1,i},Y_{2,i});
\end{equation} 
\item[(b)] combines the convexity of $S(D_2)$ in $D_2$ with Jensen's inequality; and 
\item[(c)] $S(D_2)$ is non-increasing in $D_2$ and
\begin{equation}
D_2 + \epsilon \geq \mathbb{E} \frac{1}{n} \sum_{i=1}^n \delta_2(X_i,\hat{X}_{2,i}). 
\end{equation}
\end{itemize}

Consider~\eqref{Eqn:ConverseStep1}, \eqref{Eqn:ConverseStep2} and~\eqref{Eqn:ConverseStep4}. We have 
\begin{equation}
R + \epsilon \geq H(\tilde{X} | Y_1) - \varepsilon(n,\epsilon) + S(D_2 + \epsilon) + \frac{1}{n} \Big( I(M ; \bm{Y_2} | \bm{\tilde{X}} ) - I(M ; \bm{Y_1} | \bm{\tilde{X}} ) \Big).
\end{equation}

We now apply Lemma~\ref{Lem:Conditional-Mathis-Lemma} with 
\begin{equation}
R=X, \ S_1=Y_1,\ S_2=Y_2,\ T=\emptyset,\ L = \tilde{X}\ \text{and}\ J = M.
\end{equation}
There exists $W$, jointly distributed with $(X,$ $Y_1,Y_2,\tilde{X})$, such that 
\begin{equation}
W\ \markov X\ \markov (Y_1,Y_2), 
\end{equation}
$|\set{W}| \leq |\set{X}|$, and 
\begin{equation}
R + \epsilon \geq H(\tilde{X} | Y_1) - \varepsilon(n,\epsilon) + S(D_2 + \epsilon) + I(W ; Y_2 | \tilde{X} ) - I(W ; Y_1 | \tilde{X} ).
\end{equation}
The converse proof is completed by letting $\epsilon \rightarrow 0$ and invoking the
continuity of $S(D_2)$ in $D_2$. \hfill $\blacksquare$


\section{Proof of Corollary~\ref{Cor:DeterministicDegradedSI}}\label{App:DeterministicDegradedSI}

Choose $C = \tilde{X}$ in Theorem~\ref{Thm:DegradedSI} and apply the definition of $S(D_2)$ to obtain 
\begin{equation}
R(0,D_2) \leq H(\tilde{X}|Y_1) + S(D_2).
\end{equation}
The reverse inequality can be proved using a short converse; specifically, we have
\begin{align}
%
H(M) &\geq I(\bm{X},\bm{\tilde{X}},\bm{Y_1},\bm{Y_2};M) \\[3pt]
%
&\geq I(\bm{\tilde{X}};M |\bm{Y_1})  + I(\bm{X};M | \bm{\tilde{X}},\bm{Y_1},\bm{Y_2}) \\[3pt]
%
&\step{a}{=} H(\bm{\tilde{X}} | \bm{Y_1}) - H(\bm{\tilde{X}} | M, \bm{Y_1})  + I(\bm{X};M | \bm{\tilde{X}},\bm{Y_2}) \\[1pt]
&\step{b}{\geq}  n\Big(H(\tilde{X}|Y_1) - \varepsilon(n,\epsilon) + S(D_2 + \epsilon)\Big),
\end{align}
where (a) applies $M \markov (\bm{\tilde{X}},\bm{Y_2}) \markov \bm{Y_1}$ and (b) repeats the steps in~\eqref{Eqn:tildeYFano}, \eqref{Eqn:ConverseStep4}, where $\varepsilon(n,\epsilon)$ can be chosen so that $\varepsilon(n,\epsilon) \rightarrow 0$ as $\epsilon \rightarrow 0$. \hfill $\blacksquare$


\section{Proof of Lemmas~\ref{Lem:SR-Achievability-Gen} and~\ref{Lem:SSC-Achievability-Gen}}\label{App:SR-SSC-Achievability-Gen}

Lemmas~\ref{Lem:SR-Achievability-Gen} and~\ref{Lem:SSC-Achievability-Gen} are both special cases of the next theorem.

\begin{theorem}[Thm.~1,~\cite{Timo-Aug-2011-A}]\label{Thm:TCG-Achievability}
Let $(U_{123},U_{12},U_{13},U_{23},U_1,U_2,U_3)$ be any tuple of auxiliary random variables, jointly distributed with the source $(X,Y_1,Y_2,Y_3)$, such that
\begin{enumerate}
\item[(i)] there is a Markov chain
\begin{equation}\label{Eqn:TCG-Achievability-Markov}
(Y_1,Y_2,Y_3)\ \markov X\ \markov (U_{123},U_{12},U_{13},U_{23},U_1,U_2,U_3);
\end{equation}
\item[(ii)] there exist three (deterministic) maps 
\begin{subequations}\label{Eqn:TCG-Achievability-Distortions}
\begin{equation}
\phi_j : \set{U}_j \times \set{Y}_j \longrightarrow \hat{\set{X}}_j,\quad j = 1,2,3, \\
\end{equation}
with
\begin{equation}
D_j \geq \mathbb{E}\ \delta_j\big(X,\phi_j(U_j,Y_j)\big).
\end{equation}
\end{subequations}
\end{enumerate}
Then, for each such tuple of auxiliary random variables, any rate tuple $(R_1,R_2,R_3)$ satisfying the following inequalities is achievable with distortions $(D_1,D_2,D_3)$: 
%
\begin{subequations}\label{Eqn:TCG-Achievability-Rates}
\begin{align}
\notag
R_1 \geq\ &
I(X;U_{123}) - I(U_{123};Y_1) \\
\notag
& + I(X;U_{12}|U_{123}) - I(U_{12};Y_1|U_{123})\\
\notag
& + I(X,U_{12};U_{13}|U_{123}) - I(U_{13};U_{12}Y_1|U_{123})\\
& + I(X;U_1|U_{123},U_{12},U_{13}) - I(U_1;Y_1|U_{123},U_{12},U_{13})
\end{align}
%
\begin{align}
\notag
R_1 + R_2 \geq\ &
I(X;U_{123}) - \min\big\{I(U_{123};Y_1), I(U_{123};Y_2) \big\} \\
\notag
& + I(X;U_{12}|U_{123}) - \min\big\{I(U_{12};Y_1|U_{123}),I(U_{12};Y_2|U_{123})\big\} \\
\notag
& + I(X,U_{12};U_{13}|U_{123}) - I(U_{13};U_{12},Y_1|U_{123}) \\
\notag
& + I(X,U_{12},U_{13};U_{23}|U_{123}) - I(U_{23};U_{12},Y_2|U_{123}) \\
\notag
& + I(X;U_1|U_{123},U_{12},U_{13}) - I(U_1;Y_1|U_{123},U_{12},U_{13})\\
& + I(X;U_2|U_{123},U_{12},U_{23}) - I(U_2;Y_2|U_{123},U_{12},U_{23})
\end{align}
%
\begin{align}
\notag
R_1 + R_2 + R_3 \geq\ &
I(X;U_{123}) - \min\big\{I(U_{123};Y_1), I(U_{123};Y_2),I(U_{123};Y_3) \big\} \\
\notag
& + I(X;U_{12}|U_{123}) - \min\big\{I(U_{12};Y_1|U_{123}),I(U_{12};Y_2|U_{123})\big\} \\
\notag
& + I(X,U_{12};U_{13}|U_{123}) - \min\big\{ I(U_{13};U_{12},Y_1|U_{123}), I(U_{13};Y_3|U_{123}) \big\} \\
\notag
& + I(X,U_{12},U_{13};U_{23}|U_{123}) - \min\big\{ I(U_{23};U_{12},Y_2|U_{123}), I(U_{23};U_{13},Y_3|U_{123})  \big\} \\
\notag
& + I(X;U_1|U_{123},U_{12},U_{13}) - I(U_1;Y_1|U_{123},U_{12},U_{13})\\
\notag
& + I(X;U_2|U_{123},U_{12},U_{23}) - I(U_2;Y_2|U_{123},U_{12},U_{23})\\
& + I(X;U_3|U_{123},U_{13},U_{23}) - I(U_3;Y_3|U_{123},U_{13},U_{23}).
\end{align}
\end{subequations}
\end{theorem}

\subsection{Proof of Lemma~\ref{Lem:SR-Achievability-Gen}}

Suppose that the auxiliary random variables $(A_1,A_2,A_3)$ meet the conditions of Lemma~\ref{Lem:SR-Achievability-Gen}. Consider Theorem~\ref{Thm:TCG-Achievability} with $U_{12}$ and $U_{13}$ being constants and
\begin{subequations}
\begin{align}
U_{123} &= U_1 = A_1\\
U_{23}   &= U_2 = A_2\\
U_3 &= A_3.
\end{align}
\end{subequations}
The rate constraints of~\eqref{Eqn:TCG-Achievability-Rates} now simplify to those of Lemma~\ref{Lem:SR-Achievability-Gen}. 
\hfill $\blacksquare$

\subsection{Proof of Lemma~\ref{Lem:SSC-Achievability-Gen}}

Suppose that the auxiliary random variables $(A_{12},A_1,A_2)$ meet the conditions of Lemma~\ref{Lem:SSC-Achievability-Gen}. Consider Theorem~\ref{Thm:TCG-Achievability} with infinite $D_3$, set $U_{123}$, $U_{13}$, $U_{23}$ and $U_3$ to be constants, and $U_{12} = A_{12}$, $U_1 = A_1$ and $U_2 = A_2$. The rate constraints of~\eqref{Eqn:TCG-Achievability-Rates} now simplify to those of Lemma~\ref{Lem:SSC-Achievability-Gen}. \hfill $\blacksquare$


\section{Proof of Lemma~\ref{Lem:SR-Converse}}\label{App:SR-Converse}

We have
\begin{align}
%
%
R_1 +\epsilon& \geq \frac{1}{n} H(M_1) \\
%
%
&\geq \frac{1}{n} I(\bm{\tilde{X}_1}; M_1|\bm{Y_1})\\
%
%
&\step{a}{\geq} \frac{1}{n} \big(H(\bm{\tilde{X}_1}|\bm{Y_1}) - n \varepsilon_1(n,\epsilon) \big)\\
&\step{b}{=} H(\tilde{X}_1|Y_1) - \varepsilon_1(n,\epsilon), \label{Eqn:R1bound}
\end{align}
where (a) applies Fano's inequality in the same way as~\eqref{Eqn:tildeYFano}, where $\varepsilon_1(n,\epsilon)$ can be chosen so that $\varepsilon_1(n,\epsilon) \rightarrow 0$ as $\epsilon \rightarrow 0$; and (b) follows because the pair $(\bm{\tilde{X}_1}, \bm{Y}_1)$ is i.i.d. Similarly, we have
\begin{align}
%
%
R_1 + R_2+\epsilon &\geq \frac{1}{n} H(M_1,M_2)\\
%
%
&\geq \frac{1}{n} I(\bm{\tilde{X}_1},\bm{X}; M_1 , M_2|\bm{Y_1} )\\
%
%
&= \frac{1}{n} \Big(I(\bm{\tilde{X}_1};M_1,M_2|\bm{Y_1}) + I( \bm{X}; M_1,M_2 |\bm{\tilde{X}_1},\bm{Y_1}) \Big)\\
\notag
& \step{a}{=} \frac{1}{n} \Big(I(\bm{\tilde{X}_1};M_1,M_2|\bm{Y_1}) + I( \bm{X}; M_1,M_2 |\bm{\tilde{X}_1},\bm{Y_2}) +I( \bm{Y_2}; M_1,M_2 |\bm{\tilde{X}_1}) \\
%
&\hspace{11mm} - I( \bm{Y_1}; M_1,M_2 |\bm{\tilde{X}_1})   \Big)\\
\notag
& \step{b}{=} \frac{1}{n} \Big(I(\bm{\tilde{X}_1};M_1,M_2|\bm{Y}_1) + I( \bm{\tilde{X}_2}; M_1,M_2 |\bm{\tilde{X}_1},\bm{Y_2}) +  I( \bm{X}; M_1,M_2 |\bm{\tilde{X}_1},\bm{\tilde{X}_2}, \bm{Y_2}) \\
%
%
& \hspace{11mm} +I( \bm{Y_2}; M_1,M_2 |\bm{\tilde{X}_1}) - I( \bm{Y_1}; M_1,M_2 |\bm{\tilde{X}_1}) \Big)\\
\notag
& \step{c}{\geq}  H({\tilde{X}_1}|{Y_1}) + H({\tilde{X}_2} | {\tilde{X}_1},{Y_2}) - \varepsilon_1(n,\epsilon) - \varepsilon_2(n,\epsilon) + \frac{1}{n} \Big( I( \bm{X}; M_1,M_2 |\bm{\tilde{X}_1},\bm{\tilde{X}_2}, \bm{Y_2}) \\
&  \hspace{11mm}  + I( \bm{Y_2}; M_1,M_2 |\bm{\tilde{X}_1}) - I( \bm{Y_1}; M_1,M_2 |\bm{\tilde{X}_1}) \Big)\label{Eqn:R123}\\
\notag
& \step{d}{\geq}  H({\tilde{X}_1}|{Y_1}) + H({\tilde{X}_2} | {\tilde{X}_1},{Y_2}) - \varepsilon_1(n,\epsilon) - \varepsilon_2(n,\epsilon) \\
& \hspace{11mm} + \frac{1}{n} \Big(I( \bm{Y_2}; M_1,M_2 |\bm{\tilde{X}_1})   - I( \bm{Y_1}; M_1,M_2 |\bm{\tilde{X}_1})   \Big).
 \label{Eqn:R12bound}
\end{align}
The justification for the steps leading to \eqref{Eqn:R12bound} is: 
\begin{itemize}
\item[(a)] the Markov chain $(M_1,M_2) \markov (\bm{\tilde{X}_1}, \bm{X})\markov (\bm{Y_1}, \bm{Y_2})$; 
\item[(b)] $\bm{\tilde{X}_2}$ is determined by $\bm{X}$;  
\item[(c)] exploits the fact that $(\bm{\tilde{X}_1}, \bm{\tilde{X}_2}, \bm{Y_1}, \bm{Y_2})$ is i.i.d. and applies Fano's inequality twice, in a manner similar to~\eqref{Eqn:tildeYFano}, where $\varepsilon_1(n,\epsilon)$ and  $\varepsilon_2(n,\epsilon)$ can be chosen so that they tend to 0 as $\epsilon \rightarrow 0$; and
\item[(d)]  the nonnegativity of conditional mutual information.
\end{itemize}

We now bound the sum rate $R_1+R_2+R_3$.
Notice that the steps leading to \eqref{Eqn:R123} remain valid if we replace $R_1+R_2$ by $R_1+R_2+R_3$ and the pair of messages $(M_1, M_2)$ by the triple $(M_1,M_2,M_3)$. Indeed, we have
\begin{align}
\notag
R_1  + R_2 + R_3+\epsilon
&  \geq H({\tilde{X}_1}|{Y_1}) + H({\tilde{X}_2} | {\tilde{X}_1},{Y_2}) - \varepsilon_1(n,\epsilon) - \varepsilon_2(n,\epsilon)\\
\notag
& \hspace{11mm} +\frac{1}{n}\Big(  I( \bm{X}; M_1,M_2,M_3 |\bm{\tilde{X}_1},\bm{\tilde{X}_2}, \bm{Y_2}) \\
%
%
& \hspace{11mm} + I( \bm{Y_2}; M_1,M_2,M_3 |\bm{\tilde{X}_1}) - I( \bm{Y_1}; M_1,M_2,M_3 |\bm{\tilde{X}_1})   \Big)\\%
\notag
& \step{a}{=}H({\tilde{X}_1}|{Y_1}) + H({\tilde{X}_2} | {\tilde{X}_1},{Y_2}) - \varepsilon_1(n,\epsilon) - \varepsilon_2(n,\epsilon)\\
\notag
& \hspace{11mm} +\frac{1}{n} \Big(  I( \bm{X}; M_1,M_2,M_3 |\bm{\tilde{X}_1},\bm{\tilde{X}_2}, \bm{Y_3}) \\
\notag &  \hspace{11mm} +  I( M_1,M_2,M_3;  \bm{Y_3} | \bm{\tilde{X}_1}, \bm{\tilde{X}_2}) -  I( M_1,M_2,M_3;  \bm{Y_2} | \bm{\tilde{X}_1},\bm{\tilde{X}_2})  \\
& \hspace{11mm} +I( \bm{Y_2}; M_1,M_2,M_3 |\bm{\tilde{X}_1})  - I( \bm{Y_1}; M_1,M_2,M_3 |\bm{\tilde{X}_1})   \Big)
\label{Eqn:SR-Converse-4}
\end{align}
where (a) invokes the Markov chain 
\begin{equation}
(M_1,M_2,M_3) \markov (\bm{\tilde{X}_1}, \bm{\tilde{X}_2}, \bm{X})\markov (\bm{Y_2}, \bm{Y_3}).
\end{equation}

Consider the first  conditional mutual information on the right hand side of~\eqref{Eqn:SR-Converse-4}. We have
\begin{align}
%
\frac{1}{n} I(\bm{X};M_1,M_2,M_3|\bm{\tilde{X}_1},\bm{\tilde{X}_2},\bm{Y_3}) 
&\step{a}{\geq} 
\frac{1}{n} \sum_{i=1}^n I(X_i ; M_1,M_2,M_3,Y_{3,1}^{i-1},Y_{3,i+1}^n | \tilde{X}_{1,i},\tilde{X}_{2,i},Y_{3,i})\\
%
%
&\step{b}{=} 
\frac{1}{n} \sum_{i=1}^n I(X_i ; C_i | \tilde{X}_{1,i}, \tilde{X}_{2,i}, Y_{3,i}) \\
%
%
&\step{c}{\geq} \sum_{i=1}^n S'\big(\mathbb{E} \delta_3(X_i,\hat{X}_{3,i})\big)\\
%
%
&\step{d}{\geq} S'\left(\mathbb{E} \frac{1}{n} \sum_{i=1}^n \delta_3(X_i,\tilde{X}_{3,i}) \right) \\
\label{Eqn:SR-Converse-3}
&\step{e}{\geq} S'(D_3 + \epsilon),
\end{align}
where (a) follows from the same reasoning as step (a) of~\eqref{Eqn:ConverseStep3};
in (b), we define 
\begin{equation}
C_i \triangleq \big(M_1,M_2,M_3,Y_{3,1}^{i-1},Y_{3,i+1}^n\big);
\end{equation}
and (c), (d) and (e) each follow the same reasoning as steps (a), (b) and (c) of~\eqref{Eqn:ConverseStep4} respectively. 

From~\eqref{Eqn:SR-Converse-4} and \eqref{Eqn:SR-Converse-3} we obtain:
\begin{align}
\notag
R_1 + R_2 + R_3+\epsilon &\geq H(\tilde{X}_1|Y_1) + H(\tilde{X}_2 | \tilde{X}_1, Y_2) + S'(D_3 + \epsilon) + \frac{1}{n} \Big( I(M_1,M_2,M_3;\bm{Y_3}|\bm{\tilde{X}_1},\bm{\tilde{X}_2}) \\
\notag
&\hspace{13mm} - I(M_1,M_2,M_3;\bm{Y_2}|\bm{\tilde{X}_1},\bm{\tilde{X}_2}) \Big) + \frac{1}{n} \Big( I(M_1,M_2,M_3;\bm{Y_2}|\bm{\tilde{X}_1})\\
&\hspace{13mm}  - I(M_1,M_2,M_3;\bm{Y_1}|\bm{\tilde{X}_1}) \Big) - \varepsilon_1(n,\epsilon) - \varepsilon_2(n,\epsilon). \label{Eqn:R123bound} 
\end{align}

Consider~\eqref{Eqn:R12bound} and \eqref{Eqn:R123bound}, and apply Lemma~\ref{Lem:Conditional-Mathis-Lemma} twice: once for 
\begin{equation}
R= X, \ S_1=Y_1,\ S_2=Y_2,\ T=Y_3\ \text{and}\ L =\tilde{X}_1,
\end{equation}
and once for 
\begin{equation}
R=X, \ S_1=Y_2,\ S_2=Y_3,\ T=Y_1\  \text{and}\ L =(\tilde{X}_1, \tilde{X}_2).
\end{equation}
We conclude that there exist auxiliary random variables $W_1$, $W_2$ and $W_3$ with 
\begin{equation}
|\set{W}_1|,|\set{W}_2|, |\set{W}_3| \leq |\set{X}|,
\end{equation}
and 
\begin{equation}
W_j \markov X\ \markov (Y_1, Y_2,Y_3), \qquad j = 1,2,3,
\end{equation} 
such that the rate tuple $(R_1,R_2,R_3)$ satisfies 
\begin{multline}
R_1+R_2 +\epsilon \geq H(\tilde{X}_1|Y_1)+ H(\tilde{X}_2|\tilde{X}_1,Y_2) +I(W_1; Y_2|\tilde{X}_1) - I(W_1;Y_1|\tilde{X}_1) \\ -  \varepsilon_1(n,\epsilon) - \varepsilon_2(n,\epsilon)\label{Eqn:R12bound*}
\end{multline} 
and
\begin{align}
\notag
R_1 + R_2 + R_3+\epsilon 
&\geq H(\tilde{X}_1|Y_1) + H(\tilde{X}_2 | \tilde{X}_1, Y_2) + S'(D_3 + \epsilon) - \varepsilon_2(n,\epsilon) - \varepsilon_1(n,\epsilon)  \\
&\hspace{3mm} + I(W_3;Y_3|\tilde{X}_1,\tilde{X}_2)  - I(W_3;Y_2|\tilde{X}_1,\tilde{X}_2) + I(W_2;Y_2|\tilde{X}_1) - I(W_2;Y_1|\tilde{X}_1).\label{Eqn:R123bound*}
\end{align}
The converse proof follows by \eqref{Eqn:R1bound},~\eqref{Eqn:R12bound*}, and \eqref{Eqn:R123bound*}, by letting $\epsilon \to 0$, and by the continuity of $S'(D_3)$ in $D_3$. 
\hfill $\blacksquare$


\section{Proofs of Theorem~\ref{Thm:SSC-Det-Deg-1}}\label{Proof:Thm:SSC-Det-Deg-1}

\subsection{Assertion (i)}

\emph{Achievability:} The rate constraints~\eqref{Eqn:SSC-Ach-1} reduce to~\eqref{Eqn:Proof-SSC-Det-Deg-1} upon setting $A_1 = \tilde{X}_1$ and $A_{12} = A_2 = \tilde{X}_2$ and invoking the assumptions $\tilde{X}_2 = \psi'(\tilde{X}_1)$ and $H(\tilde{X}_2|Y_1) \leq H(\tilde{X}_2|Y_2)$. 

\emph{Converse:} The lower bound on $R_1$ in~\eqref{Eqn:Proof-SSC-Det-Deg-1a} is trivial. The lower bound on the sum rate $R_1 + R_2$ in~\eqref{Eqn:Proof-SSC-Det-Deg-1b} follows by, now familiar, arguments: 
\begin{align}
R_1 + R_2 + \epsilon 
&\geq \frac{1}{n} H(M_1,M_2) \\
&\geq \frac{1}{n}I(\bm{X},\bm{\tilde{X}_2};M_1,M_2 | \bm{Y_2}) \\
&= \frac{1}{n} \Big( I(\bm{\tilde{X}_2};M_1,M_2 | \bm{Y_2}) + I(\bm{{X}};M_1,M_2 | \bm{\tilde{X}_2},\bm{Y_2}) \Big)\\
\notag
&= \frac{1}{n} \Big( I(\bm{\tilde{X}_2};M_1,M_2 | \bm{Y_2}) + I(\bm{{X}};M_1,M_2 | \bm{\tilde{X}_2},\bm{Y_1})  \\
&\hspace{10mm} + I(M_1,M_2;\bm{Y_1}|\bm{\tilde{X}_2}) - I(M_1,M_2;\bm{Y_2}|\bm{\tilde{X}_2})\Big)\\
\notag
&\step{a}{\geq}  H(\tilde{X}_2 | Y_2) + H(\tilde{X}_1 | \tilde{X}_2,Y_1) - \varepsilon(n,\epsilon)  \\
&\hspace{10mm} + \frac{1}{n} \Big(I(M_1,M_2;\bm{Y_1}|\bm{\tilde{X}_2}) - I(M_1,M_2;\bm{Y_2}|\bm{\tilde{X}_2})\Big)\\
&\step{b}{=}  H(\tilde{X}_2 | Y_2) + H(\tilde{X}_1 | \tilde{X}_2,Y_1) - \varepsilon(n,\epsilon) + I(W;Y_1|\tilde{X}_2) - I(W;Y_2|\tilde{X}_2) \\
&\step{c}{\geq}  H(\tilde{X}_2 | Y_2) + H(\tilde{X}_1 | \tilde{X}_2,Y_1) - \varepsilon(n,\epsilon),
\end{align}
where (a) applies Fano's inequality and that $\tilde{X}_1$ can be computed as a function of $X$ and  $\varepsilon(n,\epsilon) \rightarrow 0$ as $\epsilon \rightarrow 0$; (b) uses Lemma~\ref{Lem:Conditional-Mathis-Lemma}; and (c) invokes the assumption $\cln{Y_1}{Y_2}{\tilde{X}_2}$.
\hfill $\blacksquare$

\subsection{Assertion (ii)}

\emph{Achievability:} The rate constraints~\eqref{Eqn:SSC-Ach-1} reduce to~\eqref{Eqn:Proof-SSC-Det-Deg-2} upon setting $A_{12} = \tilde{X}_1$, $A_2 = \tilde{X}_2$ and $A_1 = $ constant and invoking the assumptions $\tilde{X}_1 = \psi'(\tilde{X}_2)$ and $H(\tilde{X}_1|Y_1) \leq H(\tilde{X}_1|Y_2)$. 

\emph{Converse:} The converse holds because for $j=1,2$, we have $R_j\geq H(\tilde{X}_j|Y_j)\geq 0$. 
\hfill $\blacksquare$

\bibliographystyle{IEEEtran}
\bibliography{Refs}

\end{document}